\documentclass[11pt,a4]{article}
\pdfoutput=1

\usepackage[T1]{fontenc}
\usepackage{microtype}
\usepackage{fullpage}
\usepackage{authblk}
\usepackage{amsthm}
\usepackage{amssymb}
\usepackage[english]{babel}
\usepackage[utf8]{inputenc}
\usepackage{amsmath}
\usepackage{amsfonts}
\usepackage{graphicx}
\usepackage[colorinlistoftodos]{todonotes}
\usepackage{mathtools}
\usepackage{enumerate}
\usepackage{enumitem}
\usepackage{lineno}
\usepackage{thmtools,thm-restate}
\usepackage{booktabs}
\usepackage{multirow}
\usepackage{makecell}
\usepackage{geometry}
\usepackage{color}
\definecolor{darkgreen}{rgb}{0,0.5,0}
\definecolor{darkblue}{rgb}{0,0,0.8}
\definecolor{darkred}{rgb}{0.8,0,0}

\usepackage{hyperref}
\hypersetup{
   unicode=false,          % non-Latin characters in Acrobat’s bookmarks
   colorlinks=true,        % false: boxed links; true: colored links
   linkcolor=darkred,          % color of internal links (change box color with linkbordercolor)
   citecolor=darkblue,        % color of links to bibliography
   filecolor=magenta,      % color of file links
   urlcolor=black           % color of external links
}

\newcommand\defeq{\stackrel{\mathclap{\normalfont{\mbox{\text{\tiny{def}}}}}}{=}}

\newtheorem{definition}{Definition}[section]
\newtheorem{lemma}[definition]{Lemma}
\newtheorem{theorem}[definition]{Theorem}

\newtheorem{corollary}[definition]{Corollary}

\graphicspath{{./figures/}}

\newcommand{\bigo}{\mathcal{O}}
\newcommand{\ignore}[1]{}
\newcommand{\numlies}{\ensuremath{L}}
\newcommand{\virtPhi}{\widetilde{\Phi}}
\newcommand{\virtPsi}{\widetilde{\Psi}}
\newcommand{\virtLies}[1]{\textup{virt}(#1)}

\usepackage{clrscode3e}

\usepackage[noend, noline, ruled, linesnumbered]{algorithm2e}
\SetAlFnt{\normalsize}
\DontPrintSemicolon
\SetKwComment{Comment}{$\triangleright$\ }{}
\SetKwProg{Def}{def}{:}{}
\SetArgSty{textnormal}
\SetKwRepeat{Do}{do}{while}
\SetKwComment{Comment}{$\triangleright$\ }{}

\newcommand{\w}{\omega}

\newcommand{\Ncap}[1]{N^{\cap}_{#1}}

\title{\bf A Framework for Searching in Graphs\\ in the Presence of Errors}

\author[1]{\Large Dariusz Dereniowski\thanks{Partially supported by National Science Centre (Poland) grant number 2015/17/B/ST6/01887.}}
\author[2]{Stefan Tiegel}
\author[2,3]{Przemys\l{}aw~Uzna\'nski}
\author[2]{Daniel~Wolleb-Graf}
\affil[1]{Faculty of Electronics, Telecommunications and Informatics, Gda\'{n}sk~University~of~Technology,~Poland}
\affil[2]{Department of Computer Science, 
ETH Z\"urich, Switzerland}
\affil[3]{Institute of Computer Science, Faculty of Mathematics and Computer Science, University~of~Wrocław,~Poland}
\date{}
\begin{document}
\maketitle

\setcounter{page}{0}
\thispagestyle{empty}

\begin{abstract}
We consider a problem of searching for an unknown target vertex $t$ in a (possibly edge-weighted) graph. Each \emph{vertex-query} points to a vertex $v$ and the response either admits that $v$ is the target or provides any neighbor $s$ of $v$ that lies on a shortest path from $v$ to $t$. This model has been introduced for trees by Onak and Parys [FOCS 2006] and for general graphs by Emamjomeh-Zadeh et al. [STOC 2016]. In the latter, the authors provide algorithms for the error-less case and for the independent noise model (where each query independently receives an erroneous answer with known probability $p<1/2$ and a correct one with probability $1-p$).
 
We study this problem both with adversarial errors and independent noise models. First, we show an algorithm that needs at most $\frac{\log_2 n}{1 - H(r)}$ queries in case of \emph{adversarial} errors, where the adversary is bounded with its rate of errors by a known constant $r<1/2$. Our algorithm is in fact a simplification of previous work, and our refinement lies in invoking an amortization argument. We then  show that our algorithm coupled with a Chernoff bound argument leads to a simpler algorithm for the independent noise model and has a query complexity that is both simpler and asymptotically better than the one of Emamjomeh-Zadeh et al. [STOC 2016].

Our approach has a wide range of applications. First, it improves and simplifies the Robust Interactive Learning framework proposed by Emamjomeh-Zadeh and Kempe [NIPS 2017]. Secondly, performing analogous analysis for \emph{edge-queries} (where a query to an edge $e$ returns its endpoint that is closer to the target) we actually recover (as a special case) a noisy binary search algorithm that is asymptotically optimal, matching the complexity of  Feige et al. [SIAM J. Comput. 1994]. Thirdly, we improve and simplify upon an algorithm for searching of \emph{unbounded} domains due to Aslam and Dhagat [STOC 1991].
\end{abstract}
\vfill

\clearpage

\section{Introduction} \label{sec:intro}

Consider the following game played on a simple connected graph $G=(V,E)$:
\begin{quote}
 Initially, the \emph{Responder} selects a \emph{target} $v^*\in V$.
 In each \emph{round}, the Questioner asks a \emph{vertex-query} by pointing to a vertex $v$ of $G$, and the Responder provides a \emph{reply}.
 The reply either states that $v$ is the target, i.e., $v=v^*$, or provides an edge incident to $v$ that lies on a shortest path to the target, breaking ties arbitrarily.
 A specific number of replies can be erroneous (we call them \emph{lies}).
 The goal is to design a strategy for the Questioner that identifies $v^*$ using as few queries as possible.
\end{quote}
We remark that this problem is known, among several other names, as R\'{e}nyi-Ulam games~\cite{Renyi61,Ulam76}, \emph{noisy binary search} or \emph{noisy decision trees}~\cite{FeigeRPU94,KarpK07,Ben-OrH08}.
One needs to put some restriction as how often the Responder is allowed to lie.
Following earlier works, we focus on the most natural probabilistic model, in which each reply is independently correct with a certain fixed probability.

This problem has interesting applications in noisy interactive learning~\cite{Angluin87,Emamjomeh-ZadehK17,KearnsV94,Littlestone87,2012Settles}.
In general terms, the learning process occurs as a version of the following scheme.
A user is presented with some \emph{information} --- this information reflects the current state of knowledge of the system and should take into account earlier interactions with the user (thus, the process is interactive).
Then, the user responds, which provides a new piece of data to the system. 
In order to model such dynamics as our problem, one needs to place some rules: what the information should look like and what is allowed as a valid user's response.
A crucial element in those applications is the fact that the learning process (reflected by queries and responses) does not require an explicit construction of the underlying graph on which the process takes place.
Instead, it is enough to argue that there \emph{exists} a graph whose vertices reflect possible states.
Moreover, this graph needs to have the property that a valid user's response reveals an edge lying on a shortest path to the state that needs to be determined by the system.
Specific applications pointed out in~\cite{Emamjomeh-ZadehK17} are the following.
In \emph{learning a ranking} the system aims at learning user's preference list \cite{RadlinskiJ05,Liu11}.
An information presented to the user is some list, and as a response the user swaps two consecutive elements on this list which are in the wrong order with respect to the user's target preference list.
Or, the response may reveal which element on a presented list has the highest rank.
Both versions of the response turn out to be consistent with our graph-theoretic game over a properly defined graph, whose vertex set is the set of all possible preference lists.
Another application is \emph{learning a clustering}, where the user's reply tells the system that in the current clustering some cluster needs to be split (the reply does not need to reveal how) or two clusters should be merged \cite{AwasthiBV17,BalcanB08}.
Yet another application includes learning a binary classifier.
The strength that comes from a graph-theoretic modeling of those applications as our game is that, although the underlying graph structure has usually exponential number of vertices (for learning a ranking it is $l!$, where $l$ is the maximum length of the preference list), the number of required queries is asymptotically logarithmic in this size \cite{Emamjomeh-Zadeh:2015aa,Emamjomeh-ZadehK17}.
Thus, the learning strategies derived from the algorithms in~\cite{Emamjomeh-Zadeh:2015aa} and~\cite{Emamjomeh-ZadehK17} turn out to be quite efficient.
We stress out that the lies in the query game reflect the fact that the user may sometimes provide incorrect replies.
We also note that any improvement of those algorithms, at which we aim in this work, leads to immediate improvements in the above-mentioned applications.

In~\cite{Emamjomeh-Zadeh:2015aa}, the authors provide an algorithm with the following \emph{query complexity}, i.e., the worst-case number of vertex-queries:
\begin{equation} \label{eq:EKS}
\frac{1}{1-H(p)}\left(\log_2 n+\bigo(\frac{1}{C}\log n + C^2 \log \delta^{-1})\right), \text{where } C=\max\left((\frac{1}{2}-p) \sqrt{\log \log n},1\right)
\end{equation}
that identifies the target with probability at least $1-\delta$, where $n$ is the number of vertices of an input graph and $H(p)=-p\log p-(1-p)\log(1-p)$ is the entropy and $p$ is the success probability of a query.
It is further observed that when $p<1/2$ is constant (w.r.t.~to $n$), then \eqref{eq:EKS} reduces to 
$\frac{\log_2 n}{1-H(p)} + o(\log n) + \bigo(\log^2 \delta^{-1}).$ $\frac{\log_2 n}{1-H(p)} + o(\log n) + \bigo(\log^2 \delta^{-1}).$ However, this complexity deteriorates when $\frac12 - p \ll \frac{1}{\sqrt{\log \log n}}$, and then \eqref{eq:EKS} becomes $\bigo(\frac{1}{1-H(p)}(\log n + \log \delta^{-1}))$.

\subsection{Our Contribution --- Improved Query Complexity}
In our analysis, we first focus on an \emph{adversarial} model called \emph{linearly bounded}, in which a rate of lies $r<1/2$ is given at the beginning of the game and the Responder is restricted so that at most $rt$ lies occur in a game of length $t$.
It turns out that this model is easier to analyze and leads to the following theorem whose proof is postponed to Section~\ref{sec:noisy}.
\begin{theorem}
\label{th:vertex_lin_bounded}
In the linearly bounded error model, with known error rate $r<1/2$, the target can be found in at most $\frac{\log_2 n}{1-H(r)}$ vertex queries.
\end{theorem}
This bound is strong enough to make an improvement in the probabilistic model.
By a simple application of Chernoff bound, we get the following query complexity. 
\begin{theorem} \label{th:vertex-probabilistic}
In the probabilistic error model with error probability $p<1/2$, the target can be found using at most
\[\frac{1}{1 - H(p)}\left(\log_2 n + \bigo(\sqrt{\log n \log \delta^{-1}} \cdot \log \frac{\log n}{\log \delta^{-1}}) + \bigo(\log \delta^{-1})\right)\]
vertex queries, correctly with probability at least $1-\delta$.
\end{theorem}

\paragraph{Simplifying the bound.}
For any $A,B$ it holds that $\sqrt{AB} \log \frac{A}{B} \le \sqrt{AB} \log (AB) = \bigo(A/\log A) + \bigo(B \log^3 B)$.\footnote{If $A<B$, then $\sqrt{AB} \le B$. If $B < A^{0.8}$, then $\sqrt{AB}\log(AB) = \bigo(A^{0.9}\log A) = \bigo(A/\log A)$. Otherwise, $A/\log A + B \log^3 B \ge 2 \sqrt{AB \frac{\log^3 B}{\log A}} =  \Theta(\sqrt{AB}\log(AB))$.}
We thus derive a query complexity of
\[\frac{1}{1 - H(p)}\Big(\log_2 n + o(\log n) + \bigo(\log \delta^{-1} \log^3 \log \delta^{-1})\Big).\]

\paragraph{Error comparison with \cite{Emamjomeh-Zadeh:2015aa}.}
We compare, in the independent noise model, the precise query complexities of~\cite{Emamjomeh-Zadeh:2015aa}, i.e.~\eqref{eq:EKS} with Theorem~\ref{th:vertex-probabilistic}. 
Observe that 
\begin{align*}
\log n \cdot \frac{1}{C}+ \log \delta^{-1} \cdot C^2 &\ge ( C^{-1} \log n)^{2/3} (C^2 \log \delta^{-1} )^{1/3}\\
 &=  \sqrt{\log n \log \delta^{-1}} \cdot \left(\frac{\log n}{\log \delta^{-1}}\right)^{1/6}
 \end{align*}
  and that  $\log \delta^{-1} \cdot C^2 \ge  \log \delta^{-1}$ (since $C \ge 1$).
Thus, our bound from Theorem~\ref{th:vertex-probabilistic} for all ranges of parameters asymptotically improves the one in \eqref{eq:EKS}.

Note that the compared bounds are with respect to worst-case strategy lengths. Our bounds can be made in expectation smaller by a factor of roughly $1-\delta$ using the same techniques as in  \cite{Ben-OrH08} and  \cite{Emamjomeh-Zadeh:2015aa}.

\subsection{Our Contribution --- Simplified Algorithmic Techniques}
The crucial underlying idea behind the algorithm from \cite{Emamjomeh-Zadeh:2015aa} that reaches the query complexity in~\eqref{eq:EKS} is as follows.
The algorithm maintains a weight function $\mu$ for the vertex set of the input graph $G=(V,E)$ so that, at any given time, $\mu(v)$ represents the likelihood that $v$ is the target.
Initially, all vertices have the same weight.
For a given $\mu$, define a \emph{potential} of a vertex $v$ to be $\Phi_{\mu}(v)=\sum_{u\in V}\mu(u)d(v,u)$, where $d(u,v)$ is the distance between the vertices $u$ and $v$ in $G$.
A vertex $q$ that minimizes this potential function is called a \emph{weighted median}, or a \emph{median} for short, $q=\arg\min_{v\in V}\Phi_{\mu}(v)$.
The vertex to be queried in each iteration of the algorithm is a median (ties are broken arbitrarily).
After each query, the weights are updated: the weight of each vertex that is compatible with the reply is multiplied by $p$, and the weights of the remaining vertices are multiplied by $1-p$.

The above scheme for querying subsequent vertices is the main building block of the algorithm that reaches the query complexity in~\eqref{eq:EKS}.
However, the analysis of the algorithm reveals a problematic case, namely the vertices that account for at least half of the total weight, call them \emph{heavy}.
On one side, such vertices are good candidates to include the target, so they are `removed' from the graph to be investigated later.
However, the need to investigate them in this separate way leads to an algorithm that has three phases, where the first two end by trimming the graph by leaving only the heavy vertices for the next phase.
The first two phases are sequences of vertex queries performed on a median.
The last phase uses yet a different majority technique.
The duration of each of the first two phases are dictated by complicated formulas, which makes the algorithm difficult to analyze and understand.

We propose a simpler algorithm than the one in~\cite{Emamjomeh-Zadeh:2015aa}.
In each step, we simply query a median until just one candidate target vertex remains.
Our improvement lies in a refined analysis in how such a query technique updates the weights, which has several advantages.
It not only leads to a better query complexity but also provides a much simpler proof.
Moreover, it results in a better understanding as how querying a median works in general graphs.
We point out that this technique is quite general: it can be successfully applied to other query models --- the details can be found in the appendix.

\subsection{Related Work}

Regarding the problem of searching in graphs without errors, many papers have been devoted to trees, mainly because it is a structure that naturally generalizes paths, which represents the classical binary search (see e.g.~\cite{LaberMP01} for search in a path with non-uniform query times).
This query model in case of trees is equivalent to several other problems, including vertex ranking~\cite{Dereniowski08} or tree-depth~\cite{NesetrilM06}.
%There exist linear-time algorithms for solving both vertex- and edge-query problems~\cite{LamY01,MozesOW08,OnakP06,Schaffer89}.
There exist linear-time algorithms for finding optimal query strategies~\cite{OnakP06,Schaffer89}.
A lot of effort has been done to understand the complexity for trees with non-uniform query times.
It turns out that the problem becomes hard for trees~\cite{DereniowskiN06,DereniowskiKUZ17}.
Also refer the reader to works on a closely related query game with edge queries \cite{CicaleseJLV12,CicaleseKLPV16,Dereniowski06,LamY01,MozesOW08}.
For general graphs, a strategy that always queries a 1-median (the minimizer of the sum of distances over all vertices) has length at most $\log_2 n$~\cite{Emamjomeh-Zadeh:2015aa}.

To shift our attention to searching in graphs with errors, two works have been recently published on probabilistic models \cite{Emamjomeh-Zadeh:2015aa,Emamjomeh-ZadehK17}.
These models are further generalized in~\cite{Deligkas:2017aa} by considering the case of identifying two targets $t_1$ and $t_2$, where each answer to a query gives an edge on a shortest path to $t_1$ with probability $p_1$ or to $t_2$ with probability $p_2 = 1 - p_1$, respectively.
Furthermore, there exists a closely related model in which the search is restricted in such a way, that each query performed to a vertex $v$ must be followed by a vertex query to one of its neighbors --- see~\cite{BoczkowskiKR16,HanusseIKN10,HanusseKK04,HanusseKKK08,KranakisK99} --- in this context errors are usually referred to as unreliable advice.

An extensive amount of work has been devoted to searching problems in the presence of lies in a non-graph-theoretic context.
The main tool of analysis is the concept of \emph{volume} introduced by Berlekamp~\cite{Berlekamp68} --- see also \cite{Cicalese:2013:FSA:2568443,Deppe2007} for a more detailed descriptions.
We skip references to very numerous works that deal with fixed number of lies, pointing to surveys in~\cite{Cicalese:2013:FSA:2568443,Deppe2007,PELC200271}.
For general queries, it is known~\cite{rivest1980coping} that a strategy of length $\log n + \numlies \log_2 \log_2 n + \mathcal{O}(\numlies \log \numlies)$ exists, where $n$ is the size of the search space and $\numlies$ is an upper bound on the number of lies.
An almost optimal approximation strategy can be found in~\cite{Muthukrishnan94}, which is actually given for a more general model of $q$-ary queries.
For the most relevant model in our context, the probabilistic model, we remark on the early works, which bound strategy lengths to $\bigo(\frac{1}{\text{poly}(\varepsilon)} \log n \log \delta^{-1})$, where $p < \frac{1}{2}$ and $\varepsilon=\frac{1}{2}-p$, with confidence probability $1-\delta$~\cite{aslam1995noise,BorgstromK93}.
A strategy of length $\bigo(\varepsilon^{-2}(\log n+\log \delta^{-1}))$ is given in~\cite{FeigeRPU94}.
Finally,~\cite{Ben-OrH08} gives the best known bound of $\frac{1}{1-H(p)}(\log_2 n+\bigo(\log\log n)+\bigo(\log \delta^{-1}))$.
We note that we arrive at a strategy matching asymptotically the complexity of \cite{FeigeRPU94} as a by-product from our graph-theoretic analysis (presented in the appendix).

\section{Preliminaries} \label{sec:preliminaries}
We now introduce the notation regarding the dynamics of the game. %, which is used both for vertex and edge queries.
We assume an input graph with non-uniform edge lengths, and we denote said lengths by $\w(e)$.
We denote by $d(u,v)$ the \emph{distance} between two vertices $u$ and $v$, which is the length of a shortest path in $G$ between $u$ and $v$.
We first focus on a simplified error model where the Responder is allowed a fixed number of lies, with the upper bound denoted as $\numlies$.
During the game, the Questioner keeps track of a \emph{lie counter} $\ell_v$ for each vertex $v$ of $G$.
The value of $\ell_v$ equals the number of lies that must have already occurred assuming that $v$ is actually the target $v^*$. 
The Questioner will utilize a constant $\Gamma > 1$ that will be fixed later.
The goal of having this parameter is that we can tune it in order to obtain the right asymptotics.
We define a \emph{weight} $\mu_t(v)$ of a vertex $v$ at the end of a round $t>0$:
\[\mu_t(v)=\mu_0(v)\cdot\Gamma^{-\ell_v},\]
where $\mu_0(v)$ is the initial weight of $v$.
For subsets $U\subseteq V$, let $\mu(U) = \sum_{v\in U} \mu(v)$.
For brevity we write $\mu_t$ in place of $\mu_t(V)$.
For a queried vertex $q$ and an answer $v$, a vertex $u$ is \emph{compatible} with the answer if $u=v$ when $q=v$, or $v$ lies on a shortest path from $q$ to $u$.

As soon as there is only one vertex $v$ left with $\ell_v \le \numlies$, the Questioner can successfully detect the target, $v^*=v$.
We will set the initial weight of each vertex $v$ to be $\mu_0(v)=1$.
Thus, $\mu_0 = n$ and $\mu_T \geq \Gamma^{-\numlies}$ if the strategy length is $T$.

Based on the weight function $\mu$, we define a \emph{potential} of a vertex $v$:
\begin{equation*}
\Phi(v) = \sum_{u \in V} \mu(u) \cdot  d(v,u).
\end{equation*}
We write $\Phi_t(v)$ to refer to the value of a potential at the end of round $t$.
Any vertex $x \in V$ minimizing $\Phi(x)$ is called \emph{$1$-median}.

Denote for an edge $\{v,u\}$, $N(v,u) = \{x \mid d(u,x)+\w(\{v,u\}) = d(v,x)\}$ to be the set of all vertices to which some shortest path from $v$ leads through $u$.
Thus, $N(v,u)$ consists of the compatible vertices for the answer $u$ when $v$ has been queried.
For any $S\subseteq V$, we write for brevity $\overline{S}=V\setminus S$, and for singletons $\overline{\{v\}}$ we further shorten to $\overline{v}$.
We say that a vertex $v$ is \emph{$\alpha$-heavy}, for some $0 \le \alpha \le 1$, if $\mu(v) > \alpha \cdot \mu(V)$.
For a queried vertex $q$, if the answer is $q$, then such a reply is called a \emph{yes-answer}; otherwise it is called a \emph{no-answer}.

\section{Vertex Searching} \label{sec:generic-statement}
We now formally state the search strategy for a fixed number of lies --- see Algorithm~\ref{str:graph_vertices_fixed}.
We combine our weight together with the idea of querying a $1$-median~\cite{Emamjomeh-Zadeh:2015aa}.
%Our main contribution is an amortized analysis that provides an exact bound on the number of vertex queries.
As announced earlier, it turns out that our bound together with an appropriately selected weight function are strong enough so that we do not need the additional stages enhanced with a majority selection used in~\cite{Emamjomeh-Zadeh:2015aa} in order to gain asymptotic improvements.
We also note that we can easily introduce technical modifications to this strategy by changing the initial weight, the value of $\Gamma$ or the stopping condition.
We will do this to conclude several results for various error models (see the appendix).

\begin{figure}
\begin{center}
\begin{minipage}{.7\linewidth}
\begin{algorithm}[H]
\SetAlgoRefName{VERTEX}
	\caption{Vertex queries for a fixed number of $\numlies$ lies.}
	\label{str:graph_vertices_fixed}
	\For{ $v \in V$ }
	{
		$\mu(v) \gets 1$\;
		$\ell_v \gets 0$
	}
	\While{more than one vertex $x \in V$ has $\ell_x \leq \numlies$}
	{
		$q \gets \arg \min_{x \in V} \Phi(x)$\;
		query the vertex  $q$\;
		\For{all nodes $u$ not compatible with the answer}
		{
			$\ell_u \gets \ell_u + 1$\;
			$\mu(u) \gets \mu(u)/\Gamma$
		}
	}
	\Return the only $x$ such that $\ell_x \leq \numlies$
\end{algorithm}
\end{minipage}
\end{center}
\end{figure}

\subsection{Analysis of the Strategy} \label{sec:proofs}
In this subsection we prove the following main technical contribution.
\begin{theorem}
\label{th:vertex_queries_exact}
Algorithm~\ref{str:graph_vertices_fixed} finds the target in at most
\[\frac{ 1}{\log_2(2\Gamma/(\Gamma+1))} \log_2 n + \frac{\log_2 \Gamma}{\log_2(2\Gamma/(\Gamma+1))} \cdot \numlies\]
vertex queries.
\end{theorem}

Note that, due to the values of the initial and the final weight, it is enough to argue that the weight decreases on average, i.e., in an amortized way, by a factor of $(\Gamma+1)/(2\Gamma)$ per round.
We first handle two cases (see Lemmas~\ref{lem:no} and~\ref{lem:yes}) when the weight decreases appropriately after a single query.
These cases are a no-answer, and a yes-answer but only when the queried vertex is not $1/2$-heavy.
In the remaining case, i.e., when the queried vertex $q$ is $1/2$ heavy, it is not necessarily true that the weight decreases by the desired factor --- this particularly happens in case of a yes-answer to such a query.
This case is handled by the amortized analysis: we pair such yes-answers with no-answers to the query on $q$ and show that in each such pair the weight decreases appropriately.

\begin{lemma} \label{lem:no}
If Algorithm~\ref{str:graph_vertices_fixed} receives a no-answer in a round $t+1$, then $\mu_{t+1} \leq \frac{\Gamma+1}{2\Gamma} \mu_t$.
\end{lemma}
\begin{proof}
Let $q$ be the vertex queried in round $t+1$. Assume that the reply is some neighbor $v$ of $q$.
By \cite{Emamjomeh-Zadeh:2015aa}, Lemma~4, we get that 
$\mu_t(N(q,v))\leq\mu_t/2$.
Moreover, because the lie counter increases by one for all vertices in $\overline{N(q,v)}$ and does not change for all vertices in $N(q,v)$ in round $t+1$, it follows that
$\mu_{t+1} =  \mu_t(N(q,v))+\frac{1}{\Gamma}\mu_t(\overline{N(q,v)}) \leq \frac{\Gamma+1}{2\Gamma} \mu_t.$
\end{proof}

\begin{lemma} \label{lem:yes}
Suppose that Algorithm~\ref{str:graph_vertices_fixed} queries in round $t+1$ a vertex $q$ that is not $1/2$-heavy.
If a yes-answer is received, then $\mu_{t+1} \leq \frac{\Gamma+1}{2\Gamma} \mu_t$.
\end{lemma}
\begin{proof}
The lie counter increments for each vertex of $G$ except for $q$ and remains the same for $q$ in round $t+1$: $\mu_{t+1}(q)=\mu_t(q)$ and $\mu_{t+1}(\overline{q})=\frac{1}{\Gamma}\mu_t(\overline{q})$.
Since $q$ is not $1/2$-heavy at the beginning of round $t+1$, $\mu_{t}(q)\leq\mu_{t}/2$.
Thus, we get
$\mu_{t+1} = \mu_t(q)+\frac{1}{\Gamma}\mu_t(\overline{q}) \leq  \frac{\Gamma+1}{2\Gamma} \mu_t.$
\end{proof}

Now we turn to the proof of Theorem~\ref{th:vertex_queries_exact}.
Consider a maximal interval $[t_1,t_2]$, where $t_1\leq t_2$ are integers, such that there exists a vertex $q$ that is $1/2$-heavy in each round $t_1,\ldots,t_2$, and $q$ is not $1/2$-heavy in round $t_2+1$.
Call it a \emph{$q$-interval}.
Note that $t_1>0$ and $q$ is not $1/2$-heavy in round $t_1-1$.
We permute the replies given by the Responder in the $q$-interval to obtain a \emph{new} sequence of replies as follows.
The replies in rounds $1,\ldots,t_1-1$ and $t_2+1$ onwards are the same in both sequences.
Note that in the interval $[t_1,t_2]$ the number of yes-answers, denote it by $p$, is smaller than or equal to the number of no-answers.
Reorder the replies in the $q$-interval so that the yes-answers occur in rounds $t_1+2i$ for each $i\in\{0,\ldots,p-1\}$.
In other words, we pair the yes-answers with no-answers so that a yes-answer in round $t_1+2i$ is paired with a no-answer in round $t_1+2i+1$; we call such two rounds a \emph{pair}.
Following the pairs, some remaining, if any, no-answers follow in rounds $t_1+2p,\ldots,t_2$.
Perform this transformation as long as a $q$-interval exists for some $q\in V$.
Denote by $\mu'$ the weight of the new sequence.

Denote by $t'$, if it exists, the minimum integer such that for some vertex $v$ and for each $t>t'$, $v$ is $1/2$-heavy at the end of the round $t$. 
If no such $t'$ exists, then let $t'$ be defined to be the number of rounds of the strategy.

We first analyze what happens, in the new sequence, in rounds $i$ and $i+1$ that are a pair in an arbitrary $q$-interval for some vertex $q$.
After such two rounds the lie counter for $q$ increases by one, and the lie counter of any other vertex increases by at least one.
This in particular implies that $q$ is a 1-median throughout the entire $q$-interval in the new sequence.
Moreover, the two replies in these rounds result in weight decrease by a factor of at least $\Gamma$, $\mu_{i+1}'\leq\mu_{i-1}'/\Gamma$.
Since $\frac{1}{\Gamma}<(\frac{1+\Gamma}{2\Gamma})^2$, the overall progress after the pair is as required.

We now prove that for each $t\in\{0,\ldots,t'-1\}$ that does not belong to any pair it holds
\begin{equation} \label{eq:Phi-prime}
\mu_{t+1}' \leq \frac{\Gamma+1}{2\Gamma} \mu_{t}'.
\end{equation}
Recall that for each $t\leq t'$ that does not belong to any $q$-interval, $\mu_t'(v)=\mu_t(v)$ for each $v\in V$.
If the answer to this query is a no-answer, then~\eqref{eq:Phi-prime} follows from Lemma~\ref{lem:no}.
Lemma~\ref{lem:no} also applies to no-answers of a $q$-interval that do not belong to any pair since, as argued above, $q$ is a $1$-median throughout the $q$-interval.
%Note that this also covers the queries in the reminder of any $q$-interval.
If the answer is a yes-answer, then since the queried vertex $q$ is not $1/2$-heavy due to the choice of $q$-intervals, Inequality~\eqref{eq:Phi-prime} follows from Lemma~\ref{lem:yes}.

If $t'$ is the last round in the original search strategy, then the proof is completed.
Otherwise, consider the suffix of the original sequence of replies, consisting of rounds $t$ for $t>t'$.
In all these rounds, by definition, some vertex $q$ is $1/2$-heavy.
Also by definition, both sequences $\mu$ and $\mu'$ are identical in this suffix.
One can check that if a vertex is heavy at the end of some round, then in the subsequent round Algorithm~\ref{str:graph_vertices_fixed} does query this vertex.
Thus, the vertex $q$ is queried in all rounds of the suffix, and hence $q$ is the target.
Thus, it is enough to observe how the weight decreases on $\overline{q}$ in case of a yes-answer in a round $t>t'$: $\mu_t'(\overline{q})=\mu_{t-1}'(\overline{q})/\Gamma \leq \frac{\Gamma+1}{2\Gamma} \mu_{t-1}'(\overline{q})$.
This completes the proof of Theorem~\ref{th:vertex_queries_exact}.

\subsection{Proof of Theorem~\ref{th:vertex_lin_bounded}}

We turn our attention to the model with a rate of lies bounded by a fraction $r<1/2$ (linearly bounded error model).
Our result, Theorem~\ref{th:vertex_lin_bounded}, is obtained on the basis of Algorithm~\ref{str:graph_vertices_fixed} and the precise bound from Theorem~\ref{th:vertex_queries_exact}.
In particular, we run Algorithm~\ref{str:graph_vertices_fixed} with $\Gamma = \frac{1-r}{r}$ and with a fixed bound on number of lies $\numlies = \frac{\log_2 n}{1 - H(r)}r$.
By Theorem~\ref{th:vertex_queries_exact}, Algorithm~\ref{str:graph_vertices_fixed} asks then at most
$ \frac{ \log_2 n }{\log_2 (2\cdot(1-r))}+ \frac{\log_2 \frac{1-r}{r}}{\log_2 (2\cdot (1-r))} \cdot \numlies = \frac{\log_2 n}{1-H(r)} \cdot \frac{ 1 -H(r) + r \log_2 \frac{1-r}{r}  }{1+\log_2 (1-r)} = \frac{\log_2n}{1-H(r)} = L/r$ queries. This bound concludes the proof, since the number of lies is within $r$ fraction of strategy length.

\subsection{Proof of Theorem~\ref{th:vertex-probabilistic}} \label{sec:noisy}

Let $\varepsilon>0$ be such that $p=\frac{1}{2}(1-\varepsilon)$.
We run the strategy from Theorem~\ref{th:vertex_lin_bounded} with an error rate $r = \frac{1}{2}(1-\varepsilon_0)$, where $ \varepsilon_0 = \varepsilon/ \left(1+\sqrt{2 \ln \delta^{-1}/\ln n}\right)$. 
By Theorem~\ref{th:vertex_lin_bounded} the strategy length is  $Q = \frac{\log_2 n}{1 - H(r)}$ which is (up to lower-order terms) $2 \varepsilon_0^{-2} \ln n$, but we can safely lowerbound it as $\varepsilon_0^{-2} \ln n$.
The expected number of lies is $\mathbb{E}[L] =p \cdot Q$ and by the Hoeffding bound, 
$$\text{Pr}[Q-L \le (1-r) \cdot Q ] \le \exp\left(-\frac{1}{2} (r-p)^2  \frac{\ln n}{\varepsilon_0^2}\right) \le \exp\left(-\frac12\left(\frac{\varepsilon-\varepsilon_0}{\varepsilon_0}\right)^2 \ln n\right) = \delta.$$

\paragraph{Asymptotic properties of entropy function.}
We now proceed to bound $\frac{1-H(p)}{1-H(r)}$. For this we denote $F(x) = 1 - H(\frac12(1-x))$, and denote $\alpha = \frac{1}{1+\sqrt{2 \ln \delta^{-1}/\ln n}}$. So our goal is in fact to bound  $\frac{F(\varepsilon)}{F(\alpha \cdot \varepsilon)}$.

\begin{lemma}
For any $-1 \le x \le 1$ and $\alpha<1$ there is
$$\frac{F(x)}{F(\alpha x)} \le \frac{1}{F(\alpha)}.$$
\end{lemma}
\begin{proof}
Consider $G(x) = \ln F(\exp(x))$. It can be verified with calculus that $G''(x) \ge 0$. The claim is equivalent to
$$G(\ln x) - G(\ln \alpha + \ln x) \le G(0) - G(\ln \alpha )$$
which follows from the convexity of $G(x)$.
\end{proof}

First assume $\alpha \ge 1/2$, so $\ln n \ge \ln \delta^{-1}$. Denote $\eta = 1 - \alpha$. We observe that $\eta= \bigo(\sqrt{\frac{\ln \delta^{-1}}{\ln n}})$. We take Taylor expansion of $1/F(x)$ around $x=1$, and we have that 
$\frac{1}{F(\alpha)} = 1 + \bigo(\eta \ln \eta^{-1}).$ In this case the bound is
$$Q \le \frac{\log_2 n}{1-H(p)} \cdot \frac{1}{F(\alpha)} =  \frac{\log_2 n + \bigo(\sqrt{\ln \delta^{-1} \ln n}  \cdot (1+\ln \frac{\ln n}{\ln \delta^{-1}}))}{1-H(p)}.$$

In the second case when $\alpha \le 1/2$, from Taylor expansion around $x=0$ there is  $\frac{1}{F(\alpha)} = \Theta(\alpha^{-2})$. So in this case the bound is
$$Q \le \frac{\log_2 n}{1-H(p)} \cdot \frac{1}{F(\alpha)} = \frac{\bigo(\ln \delta^{-1})}{1-H(p)}.$$\qed
%\end{proof}

\section{Conclusions}
We note that also other query models have been studied in the graph-theoretic context, including edge queries.
In an \emph{edge query}, the Questioner points to an edge and the Responder tells which endpoint of that edge is closer to the target, breaking ties arbitrarily.
It turns out that edge queries are more challenging to analyze, i.e., our technique for vertex queries does not transfer without changes.
This is mostly due to a possible lack of edges that subdivide the search space equally enough.
This issue can be patched by treating heavy vertices in a separate way.
We provide a strategy of query complexity $\bigo(\frac{1}{\varepsilon^2} \Delta \log \Delta (\log n + \log \delta^{-1}))$.
This  generalizes the noisy binary search of \cite{FeigeRPU94} to general graphs, and has the advantage of being a weight-based strategy.

We additionally show the generalizations of our strategies to searching in unbounded domains, where one is concerned in searching e.g., the space of all positive integers with comparison queries.
The goal is to minimize the number of queries as a function of $N$, the (unknown) position of the target. By adjusting the initial distribution of the weight to decrease polynomially with respect to the distance from the point 0, we almost automatically get desired solutions for adversarial models. For probabilistic error model we present (a slightly more involved) strategy of expected query complexity $\bigo(\frac{1}{\varepsilon^2} (\log N + \log \delta^{-1}))$, improving over the complexity $\bigo(\text{poly}(\varepsilon^{-1})\log N \log{\delta^{-1}})$ in \cite{aslam1995noise}.

\bibliographystyle{alpha}
\bibliography{bib}

\newcommand{\etalchar}[1]{$^{#1}$}
\begin{thebibliography}{RMK{\etalchar{+}}80}

\bibitem[ABV17]{AwasthiBV17}
Pranjal Awasthi, Maria{-}Florina Balcan, and Konstantin Voevodski.
\newblock Local algorithms for interactive clustering.
\newblock {\em Journal of Machine Learning Research}, 18:3:1--3:35, 2017.

\bibitem[AD91]{AslamD91}
Javed~A. Aslam and Aditi Dhagat.
\newblock Searching in the presence of linearly bounded errors (extended
  abstract).
\newblock In {\em STOC}, pages 486--493, 1991.

\bibitem[Aig96]{Aigner96}
Martin Aigner.
\newblock Searching with lies.
\newblock {\em J. Comb. Theory, Ser. {A}}, 74(1):43--56, 1996.

\bibitem[Ang87]{Angluin87}
Dana Angluin.
\newblock Queries and concept learning.
\newblock {\em Machine Learning}, 2(4):319--342, 1987.

\bibitem[Asl95]{aslam1995noise}
Javed~A Aslam.
\newblock {\em Noise tolerant algorithms for learning and searching}.
\newblock PhD thesis, Massachusetts Institute of Technology, 1995.

\bibitem[BB08]{BalcanB08}
Maria{-}Florina Balcan and Avrim Blum.
\newblock Clustering with interactive feedback.
\newblock In {\em {ALT}}, pages 316--328, 2008.

\bibitem[Ber68]{Berlekamp68}
Elvyn~R. Berlekamp.
\newblock Block coding for the binary symmetric channel with noiseless,
  delayless feedback.
\newblock In {\em {H.B. Mann (ed.)}, Error-Correcting Codes}, pages 61--88.
  Wiley \& Sons, New York, 1968.

\bibitem[BH08]{Ben-OrH08}
Michael Ben{-}Or and Avinatan Hassidim.
\newblock The bayesian learner is optimal for noisy binary search (and pretty
  good for quantum as well).
\newblock In {\em FOCS}, pages 221--230, 2008.

\bibitem[BK93]{BorgstromK93}
Ryan~S. Borgstrom and S.~Rao Kosaraju.
\newblock Comparison-based search in the presence of errors.
\newblock In {\em STOC}, pages 130--136, 1993.

\bibitem[BKR18]{BoczkowskiKR16}
Lucas Boczkowski, Amos Korman, and Yoav Rodeh.
\newblock Searching a tree with permanently noisy advice.
\newblock In {\em {ESA}}, pages 54:1--54:13, 2018.

\bibitem[Cic13]{Cicalese:2013:FSA:2568443}
Ferdinando Cicalese.
\newblock {\em Fault-Tolerant Search Algorithms: Reliable Computation with
  Unreliable Information}.
\newblock Springer Publishing Company, Incorporated, 2013.

\bibitem[CJLV12]{CicaleseJLV12}
Ferdinando Cicalese, Tobias Jacobs, Eduardo~Sany Laber, and Caio~Dias Valentim.
\newblock The binary identification problem for weighted trees.
\newblock {\em Theor. Comput. Sci.}, 459:100--112, 2012.

\bibitem[CKL{\etalchar{+}}16]{CicaleseKLPV16}
Ferdinando Cicalese, Bal{\'{a}}zs Keszegh, Bernard Lidick{\'{y}},
  D{\"{o}}m{\"{o}}t{\"{o}}r P{\'{a}}lv{\"{o}}lgyi, and Tom{\'{a}}s Valla.
\newblock On the tree search problem with non-uniform costs.
\newblock {\em Theor. Comput. Sci.}, 647:22--32, 2016.

\bibitem[Dep07]{Deppe2007}
Christian Deppe.
\newblock Coding with feedback and searching with lies.
\newblock In Imre Csisz{\'a}r, Gyula O.~H. Katona, G{\'a}bor Tardos, and
  G{\'a}bor Wiener, editors, {\em Entropy, Search, Complexity}, pages 27--70.
  Springer Berlin Heidelberg, Berlin, Heidelberg, 2007.

\bibitem[Der06]{Dereniowski06}
Dariusz Dereniowski.
\newblock Edge ranking of weighted trees.
\newblock {\em Discrete Applied Mathematics}, 154(8):1198--1209, 2006.

\bibitem[Der08]{Dereniowski08}
Dariusz Dereniowski.
\newblock Edge ranking and searching in partial orders.
\newblock {\em Discrete Applied Mathematics}, 156(13):2493--2500, 2008.

\bibitem[DGW92]{DhagatGW92}
Aditi Dhagat, P{\'{e}}ter G{\'{a}}cs, and Peter Winkler.
\newblock On playing "twenty questions" with a liar.
\newblock In {\em SODA}, pages 16--22, 1992.

\bibitem[DKUZ17]{DereniowskiKUZ17}
Dariusz Dereniowski, Adrian Kosowski, Przemyslaw Uznański, and Mengchuan Zou.
\newblock Approximation strategies for generalized binary search in weighted
  trees.
\newblock In {\em ICALP}, pages 84:1--84:14, 2017.

\bibitem[DMS17]{Deligkas:2017aa}
Argyrios Deligkas, George~B. Mertzios, and Paul~G. Spirakis.
\newblock Binary search in graphs revisited.
\newblock In {\em MFCS}, pages 20:1--20:14, 2017.

\bibitem[DN06]{DereniowskiN06}
Dariusz Dereniowski and Adam Nadolski.
\newblock Vertex rankings of chordal graphs and weighted trees.
\newblock {\em Inf. Process. Lett.}, 98(3):96--100, 2006.

\bibitem[EK17]{Emamjomeh-ZadehK17}
Ehsan Emamjomeh{-}Zadeh and David Kempe.
\newblock A general framework for robust interactive learning.
\newblock In {\em {NIPS}}, pages 7085--7094, 2017.

\bibitem[EKS16]{Emamjomeh-Zadeh:2015aa}
Ehsan Emamjomeh{-}Zadeh, David Kempe, and Vikrant Singhal.
\newblock Deterministic and probabilistic binary search in graphs.
\newblock In {\em STOC}, pages 519--532, 2016.

\bibitem[FRPU94]{FeigeRPU94}
Uriel Feige, Prabhakar Raghavan, David Peleg, and Eli Upfal.
\newblock Computing with noisy information.
\newblock {\em {SIAM} J. Comput.}, 23(5):1001--1018, 1994.

\bibitem[HIKN10]{HanusseIKN10}
Nicolas Hanusse, David Ilcinkas, Adrian Kosowski, and Nicolas Nisse.
\newblock Locating a target with an agent guided by unreliable local advice:
  how to beat the random walk when you have a clock?
\newblock In {\em PODC}, pages 355--364, 2010.

\bibitem[HKK04]{HanusseKK04}
Nicolas Hanusse, Evangelos Kranakis, and Danny Krizanc.
\newblock Searching with mobile agents in networks with liars.
\newblock {\em Discrete Applied Mathematics}, 137(1):69--85, 2004.

\bibitem[HKKK08]{HanusseKKK08}
Nicolas Hanusse, Dimitris~J. Kavvadias, Evangelos Kranakis, and Danny Krizanc.
\newblock Memoryless search algorithms in a network with faulty advice.
\newblock {\em Theor. Comput. Sci.}, 402(2-3):190--198, 2008.

\bibitem[KK99]{KranakisK99}
Evangelos Kranakis and Danny Krizanc.
\newblock Searching with uncertainty.
\newblock In {\em {SIROCCO}}, pages 194--203, 1999.

\bibitem[KK07]{KarpK07}
Richard~M. Karp and Robert Kleinberg.
\newblock Noisy binary search and its applications.
\newblock In {\em SODA}, pages 881--890, 2007.

\bibitem[KV94]{KearnsV94}
Michael~J. Kearns and Umesh~V. Vazirani.
\newblock {\em An Introduction to Computational Learning Theory}.
\newblock {MIT} Press, 1994.

\bibitem[Lit87]{Littlestone87}
Nick Littlestone.
\newblock Learning quickly when irrelevant attributes abound: {A} new
  linear-threshold algorithm.
\newblock {\em Machine Learning}, 2(4):285--318, 1987.

\bibitem[Liu11]{Liu11}
Tie{-}Yan Liu.
\newblock {\em Learning to Rank for Information Retrieval}.
\newblock Springer, 2011.

\bibitem[LMP01]{LaberMP01}
Eduardo~Sany Laber, Ruy~Luiz Milidi{\'{u}}, and Artur~Alves Pessoa.
\newblock On binary searching with non-uniform costs.
\newblock In {\em SODA}, pages 855--864, 2001.

\bibitem[Lon92]{Long92}
Philip~M. Long.
\newblock Sorting and searching with a faulty comparison oracle.
\newblock Technical report, Technical Report UCSC-CRL-92-15, University of
  California at Santa Cruz, 1992.

\bibitem[LY01]{LamY01}
Tak~Wah Lam and Fung~Ling Yue.
\newblock Optimal edge ranking of trees in linear time.
\newblock {\em Algorithmica}, 30(1):12--33, 2001.

\bibitem[MOW08]{MozesOW08}
Shay Mozes, Krzysztof Onak, and Oren Weimann.
\newblock Finding an optimal tree searching strategy in linear time.
\newblock In {\em SODA}, pages 1096--1105, 2008.

\bibitem[Mut94]{Muthukrishnan94}
S.~Muthukrishnan.
\newblock On optimal strategies for searching in presence of errors.
\newblock In {\em SODA}, pages 680--689, 1994.

\bibitem[NdM06]{NesetrilM06}
Jaroslav Nesetril and Patrice~Ossona de~Mendez.
\newblock Tree-depth, subgraph coloring and homomorphism bounds.
\newblock {\em Eur. J. Comb.}, 27(6):1022--1041, 2006.

\bibitem[OP06]{OnakP06}
Krzysztof Onak and Pawel Parys.
\newblock Generalization of binary search: Searching in trees and forest-like
  partial orders.
\newblock In {\em FOCS}, pages 379--388, 2006.

\bibitem[Pel89]{Pelc89}
Andrzej Pelc.
\newblock Searching with known error probability.
\newblock {\em Theor. Comput. Sci.}, 63(2):185--202, 1989.

\bibitem[Pel02]{PELC200271}
Andrzej Pelc.
\newblock Searching games with errors---fifty years of coping with liars.
\newblock {\em Theoretical Computer Science}, 270(1):71 -- 109, 2002.

\bibitem[RJ05]{RadlinskiJ05}
Filip Radlinski and Thorsten Joachims.
\newblock Query chains: learning to rank from implicit feedback.
\newblock In {\em {SIGKDD}}, pages 239--248, 2005.

\bibitem[RMK{\etalchar{+}}80]{rivest1980coping}
Ronald~L. Rivest, Albert~R. Meyer, Daniel~J. Kleitman, Karl Winklmann, and Joel
  Spencer.
\newblock Coping with errors in binary search procedures.
\newblock {\em Journal of Computer and System Sciences}, 20(3):396--404, 1980.

\bibitem[Ré61]{Renyi61}
Alfréd Rényi.
\newblock On a problem of information theory.
\newblock {\em MTA Mat. Kut. Int. Kozl.}, 6B:505--516, 1961.

\bibitem[Sch89]{Schaffer89}
Alejandro~A. Sch{\"{a}}ffer.
\newblock Optimal node ranking of trees in linear time.
\newblock {\em Inf. Process. Lett.}, 33(2):91--96, 1989.

\bibitem[Set12]{2012Settles}
Burr Settles.
\newblock {\em Active Learning}.
\newblock Synthesis Lectures on Artificial Intelligence and Machine Learning.
  Morgan {\&} Claypool Publishers, 2012.

\bibitem[Ula76]{Ulam76}
Stanislaw~M. Ulam.
\newblock {\em Adventures of a Mathematician}.
\newblock Scribner, New York, 1976.

\end{thebibliography}

\appendix
\section{More Searching Models}

We recall a different format of queries called \emph{edge queries}, where in each round the Questioner selects an edge $\{u,v\}$ of an input graph and the Responder replies with the endpoint of $\{u,v\}$ that is closer to the target.
Again, ties are broken arbitrarily.
The edge-query model naturally generalizes comparison queries in linearly or partially ordered data.
In case of edge-queries we consider unweighted graphs.

\medskip
For the limitations imposed on the Responder, we distinguish yet another model called \emph{prefix-bounded}.
In this model, in \emph{each} prefix of $i$ queries there may be at most $ri$ lies, $0\leq r<\frac{1}{2}$, and in such case as opposed to the linearly bounded model, the length of the strategy does not need to be initially set.
It is well known that these error models are not feasible for $r\geq\frac{1}{2}$, even in the case of paths.
They bridge the gap between the adversarial one with a fixed number of lies and the probabilistic model.
We note that these models naturally reflect processes in potential applications like communication scenarios over a noisy channel or hardware errors.
This is due to the fact that in such scenarios the errors typically accumulate over time.

In the following sections we show that our generic ideas can be applied to several other models.
Our results either match or improve the existing ones, which we point out throughout.
We remark that in both cases, i.e., whether we obtain an improvement or arrive at an existing result, we reach that point with a simpler analysis.

\section{Analysis of the Generic Strategies for Edge Queries}
We start by giving the notation regarding edge queries.
The \emph{degree} of a vertex $v$, denoted by $\deg(v)$, is the number of its neighbors in $G$.
We denote by $\Delta=\max_{v\in V}\deg(v)$ the maximum degree of $G$.
We define an edge-vertex distance $d(e,v) = \min(d(x,v),d(y,v))$ for an edge $e=\{x,y\}$.
Similarly as for vertex queries, based on a weight function $\mu$ and distance $d$, we define a potential of an edge $e$:
\begin{equation*}
\Phi(e) = \sum_{u \in V} \mu(u) \cdot d(e,u).
\end{equation*}
Again, we write $\Phi_t(e)$ to refer to this value at the end of round $t$.
Any edge $e$ minimizing $\Phi(e)$ is called \emph{$1$-edge-median}.
For an edge $e = \{u,v\}$ and one of its endpoints,
\begin{align*}
&N(e,v) = \{w \mid d(v,w) \leq d(u,w)\}, &N_{<}(e,v) = \{w \mid d(v,w) < d(u,w)\}.
\end{align*}

For edge-queries we give a strategy that is a bit more complicated --- see Algorithm~\ref{str:graph_edges_fixed}.
Intuitively, as opposed to the vertex-query case, there may be no edges in the graph that `subdivide' the search space evenly enough.
This already happens as soon as one of the vertices is $\frac{1}{\Delta+1}$-heavy.
If this is the case, and say vertex $v$ is $\frac{1}{\Delta+1}$-heavy, we cyclically query edges incident to $v$ in an appropriate greedy order.
We continue to do so until all other vertices have been eliminated, and hence $v$ must be the target, or $v$ is no longer $\frac{1}{\Delta+1}$-heavy.
If none of the vertices is $\frac{1}{\Delta+1}$-heavy, we simply query a 1-edge-median.
The absence of such heavy vertices essentially ensures, that this decreases the weight sufficiently.

\begin{figure}[t]
\begin{center}
\begin{minipage}{.9\linewidth}
\begin{algorithm}[H]
\SetAlgoRefName{EDGE}
	\caption{Edge queries for fixed number of $\numlies$ lies.}
	\label{str:graph_edges_fixed}
	\For{ $v \in V$ }
	{
		$\mu(v) \gets 1$\;
		$\ell_v \gets 0$
	}
	\While{more than one vertex $x$ satisfies $\ell_x \leq \numlies$}
	{
		\If(\Comment*[f]{$v$ is $\frac{1}{\Delta+1}$-heavy}){there exists $v$ such that $\mu(v)> \mu/(\Delta+1)$}
		{
			\For(\Comment*[f]{greedy ordering of neighbors}){ $i=1$ \KwTo $\deg(v)$}
			{
				select an edge $e_i$ incident to $v$ to maximize $\mu(\bigcup_{j \le i} N_{<}(e_j,v) )$
			}
			$i \gets 1$

			\Do(\Comment*[f]{cyclically query edges incident to $v$}){$\mu(v) > \mu/(\Delta+1)$ and there exists more than one $x$ with $\ell_x \leq \numlies$}
			{
				query $e_i$\;
				\For{all nodes $u$ not compatible with the answer}
				{
					$\ell_u \gets \ell_u + 1$\;
					$\mu(u) \gets \mu(u)/\Gamma$
				}
				\If{the answer to the last query is $v$}
				{
					$i \gets (i+1)\bmod\deg(v)$\;
				}
			}
		
		}
		\Else
		{
			$e \gets \arg \min_{x \in E} \Phi(x)$\;
			query $e$\;
			\For{all nodes $u$ not compatible with the answer}
			{
				$\ell_u \gets \ell_u + 1$\;
				$\mu(u) \gets \mu(u)/\Gamma$
			}
		}
	}
	\Return $v$ such that $\ell_v \leq \numlies$
\end{algorithm}
\end{minipage}
\end{center}
\end{figure}

This results in a more involved proof given in Section~\ref{sec:proofs}.
Similarly as for vertex queries, we also first provide an analysis for a fixed number of lies (see Theorem~\ref{th:edges_queries_exact}) and then from this bound we derive appropriate bounds for other models (Theorems~\ref{th:edge_bounded} and~\ref{th:edge-probabilistic}).
\begin{theorem}
\label{th:edges_queries_exact}
Let $\Gamma>1$.
Algorithm~\ref{str:graph_edges_fixed} finds the target in at most $\frac{\log n + \numlies \log \Gamma }{\log(1+ \frac{\Gamma-1}{\Gamma\Delta+1})}$ edge queries.
\end{theorem}
\begin{theorem}
\label{th:edge_bounded}
In the linearly bounded error model with error rate $r = \frac{1}{\Delta+1}(1 - \varepsilon)$ for some $0<\varepsilon\leq 1$, the target can be found in at most $2 \varepsilon^{-2} \Delta \ln n $ edge queries.
\end{theorem}

\begin{theorem} \label{th:edge-probabilistic}
In the probabilistic error model with error rate  $p=\frac{1}{2}(1-\varepsilon)$ for some $0<\varepsilon\leq 1$ there exists a strategy that finds the target  using at most $\bigo(\varepsilon^{-2}  \Delta \log \Delta \cdot (\log n + \log \delta^{-1}) )$ edge queries, correctly with probability at least $1 - \delta$.
\end{theorem}
%We also observe that using e.g. $\Gamma = 2$ in Algorithm~\ref{str:graph_edges_fixed} yields a strategy of length $\bigo(\Delta(\numlies + \log n))$ for a fixed number of $\numlies$ lies.

\subsection*{Proof of Theorem~\ref{th:edges_queries_exact}}
\label{sec:edge-queries-with-lies}
We first prove two technical lemmas and then we give the proof of the theorem.
\begin{lemma} \label{lem:degBound}
Let $\Gamma>1$.
Suppose that Algorithm~\ref{str:graph_edges_fixed} queries in round $t+1$ an edge $e_q$ incident to a vertex $q$ such that $e_q = \arg\min_{x \in E} \Phi_{t}(x)$.
If $\deg(q)>1$, then
\begin{equation} \label{eq:Phi-deg}
\mu_t(\overline{N(e_q,q)}) \geq \frac{1}{\deg(q)} ( \mu_t - \mu_{t}(q) ).
\end{equation}
\end{lemma}
\begin{proof}
Denote $e_q=\{q,v\}$.
For each neighbor $w$ of $q$ define
\[\Ncap{w}=N(e_q,q) \cap N_{<}(\{q,w\},w).\]
Consider an edge $e' = \{q,w\}$ that maximizes $\mu_{t}(\Ncap{w})$.
If $X$ is the set of neighbors of $q$, then by definition and by the fact that $e_q$ lies on no shortest path from $q$ to any vertex in $N_{<}(e_q,v)$, i.e., $\Ncap{v}=\emptyset$, it holds
\[N(e_q,q)\setminus\{q\} \subseteq \bigcup_{w'\in X}\Ncap{w'} = \bigcup_{w'\in X\setminus\{v\}} \Ncap{w'}.\]
Hence (since $e'$ maximizes $\mu_t(\Ncap{w})$) we obtain that
\begin{equation} \label{eq:Ncap}
\mu_{t}(\Ncap{w}) \geq  \frac{1}{\deg(q) - 1} ( \mu_{t}(N(e_q,q)) - \mu_{t}(q) ).
\end{equation}

For brevity, we extend our notation in the following way: for an edge $e$ and a subset $S$ of vertices, $\Phi_t(e,S)=\sum_{z\in S}\mu_t(z)\cdot d(e,z)$.
Note that for any $S\subseteq V$ and any edge $e$, $\Phi_t(e)=\Phi_t(e,S)+\Phi_t(e,\overline{S})$.
We obtain
\begin{align} \label{eq:Psi-e-prime}
\begin{split}
\Phi_t (e', N(e_q,q)) &= \Phi_t(e',\Ncap{w}) +\Phi_t(e', N(e_q,q) \setminus \Ncap{w})
\\ &= \sum_{u \in \Ncap{w}} \mu_{t}(u) \cdot (d(q,u) - 1) + \sum_{u \in N(e_q,q) \setminus \Ncap{w}} \mu_{t}(u) \cdot d(q,u)
\\ &= \sum_{u \in N(e_q,q)} \mu_{t}(u) \cdot d(q,u) - \mu_{t}(\Ncap{w})
\\ &\le \Phi_t (e_q, N(e_q,q)) - \frac{1}{\deg(q) - 1} ( \mu_{t}(N(e_q,q)) - \mu_{t}(q) ),
\end{split}
\end{align}
where the last inequality is due to~\eqref{eq:Ncap}.
For any vertex $u$, $d(e',u) \leq d(e_q,u) + 1$ because $e_q$ and $e'$ are adjacent.
Using this fact we obtain: 
\begin{align} \label{eq:Psi-e-prime2}
\begin{split}
\Phi_t (e', \overline{N(e_q,q)}) &= \sum_{u \notin N(e_q,q)} \mu_{t}(u) \cdot d(e',u)
\\ &\leq \sum_{u \notin N(e_q,q)} \mu_{t}(u) \cdot d(e_q,u) + \sum_{u \notin N(e_q,q)} \mu_{t}(u)
\\ &= \Phi_t (e_q, \overline{N(e_q,q)}) + \mu_{t}(\overline{N(e_q,q)}).
\end{split}
\end{align}
Finally, by~\eqref{eq:Psi-e-prime} and~\eqref{eq:Psi-e-prime2} we get:
\begin{align*}
\Phi_t(e') &= \Phi_t (e', N(e_q,q)) + \Phi_t (e', \overline{N(e_q,q)})
\\ &\leq \Phi_t (e_q, N(e_q,q)) - \frac{ \mu_{t}(N(e_q,q)) - \mu_{t}(q) }{\deg(q) - 1} + \Phi_t (e_q, \overline{N(e_q,q)}) + \mu_{t}(\overline{N(e_q,q)}) 
\\ &= \Phi_t(e_q) + \mu_{t}(\overline{N(e_q,q)}) -  \frac{1}{\deg(q) - 1} ( \mu_{t}(N(e_q,q)) - \mu_{t}(q) ).
\end{align*}
By assumption, $\Phi_t(e_q)\leq\Phi_t(e')$.
Therefore,
\[\frac{1}{\deg(q) - 1} ( \mu_{t}(N(e_q,q)) - \mu_{t}(q) ) \leq \mu_{t}(\overline{N(e_q,q)}),\]
which can be rewritten as in~\eqref{eq:Phi-deg}.
\end{proof}

\begin{lemma} \label{lem:edge-no-answer}
Let $\Gamma>1$.
Suppose that Algorithm~\ref{str:graph_edges_fixed} queries in round $t+1$ an edge incident to a vertex $q$ that is not $\frac{1}{\Delta+1}$-heavy in this round, and the answer is $q$.
Then, $\mu_{t+1}\leq (1- \frac{\Gamma-1}{\Gamma(\Delta+1)})\mu_{t}$.
\end{lemma}
\begin{proof}
Let $e_q=\{q,v\}$ be the edge queried in round $t+1$.
Suppose first that $\deg(q)>1$.
By Lemma~\ref{lem:degBound},
\begin{equation} \label{eq:Phi-overline}
\mu_t(\overline{N(e_q,q)}) \geq \frac{1}{\deg(q)} ( \mu_t - \mu_{t}(q) ) \geq \frac{1}{\Delta} ( \mu_t - \mu_{t}(q) ).
\end{equation}
Because $e_q$ is the queried edge in round $t+1$ and the reply is $q$, the lie counter remains unchanged for the vertices in $N(e_q,q)$ and decreases by one in the complement $\overline{N(e_q,q)}$.
Hence,
\[\mu_{t+1} = \mu_t(N(e_q,q))+\frac{1}{\Gamma}\mu_t(\overline{N(e_q,q)}) = \mu_t - \frac{\Gamma-1}{\Gamma} \mu_{t}(\overline{N(e_q,q)}).\]
Thus, by~\eqref{eq:Phi-overline} and by the fact that $\mu_t(q)\leq\frac{1}{\Delta+1}\mu_t$ for $q$ that is not $\frac{1}{\Delta+1}$-heavy in round $t$,
\[\mu_{t+1} \leq \left( 1 - \frac{\Gamma-1}{\Gamma\Delta}\cdot\frac{\Delta}{\Delta+1} \right) \mu_t,\]
which completes the proof in the case when $\deg(q)>1$.

If $\deg(q)=1$, then in round $t$ the lie counter increases by one for each vertex in $\overline{q}$.
Thus, again by the fact that $q$ is not $\frac{1}{\Delta+1}$-heavy,
\[\mu_{t+1} = \mu_{t}(q) + \frac{1}{\Gamma}\mu_{t}(\overline{q}) \leq \left( \frac{1}{\Delta+1} + \frac{1}{\Gamma} \right) \mu_{t} \le \left(1 - \frac{\Gamma-1}{\Gamma(\Delta+1)}\right)\mu_t. \qedhere\]
\end{proof}

Having proved the technical lemmas, we now turn to the proof of Theorem~\ref{th:edges_queries_exact}.
It is enough to argue that every query, amortized, multiplies the weight by a factor of $1- \frac{\Gamma-1}{\Gamma(\Delta+1)} = 1/(1 + \frac{\Gamma-1}{\Gamma\Delta+1})$.
If there is no $\frac{1}{\Delta+1}$-heavy vertex, then the theorem follows from Lemma~\ref{lem:edge-no-answer}.
Hence suppose in the rest of the proof that there exists a $\frac{1}{\Delta+1}$-heavy vertex and denote this vertex by $q$.

For the amortized analysis, consider a sequence of $t$ consecutive queries to edges $e_1,\ldots,e_{t}$, $t\leq\deg(q)$, performed while $q$ is $\frac{1}{\Delta+1}$-heavy; call such a sequence a \emph{segment}.
Suppose this sequence starts in round $t'$.
Denote $e_i=\{q,v_i\}$, $i\in\{1,\ldots,t\}$, and let
\begin{align*}
&Q_1=\bigcup_{i=1}^t N_<(e_i,v_i), &Q_2=V\setminus(Q_1\cup\{q\}).
\end{align*}

First we assume that the query in round $t'+t$ (i.e., the query that follows the sequence) does not return $v$ as a reply, or $v$ stops being $\frac{1}{\Delta+1}$-heavy.
We argue, informally speaking, that this query in round $t'+t$ amortizes the $t$ queries prior to it thanks to the assumption $t \le \deg(q)$.
Because the lie counter of $q$ increments in round $t'+t$,
\begin{equation} \label{eq:q1}
 \mu_{t'+t}(q) \leq \frac{1}{\Gamma}\mu_{t'}(q).
\end{equation}
We have $\mu_{t'+t}(Q_1) \leq \frac{1}{\Gamma}\mu_{t'}(Q_1)$ by the formulation of Algorithm~\ref{str:graph_edges_fixed}, and $\mu_{t'+t}(Q_2) \leq \mu_{t'}(Q_2)$.
Then, $Q_1\cup Q_2=\overline{q}$ and $Q_1\cap Q_2=\emptyset$ imply $\mu_{t'}(Q_1)\leq\mu_{t'}(\overline{q})-\mu_{t'}(Q_2)$ and hence
\begin{equation} \label{eq:Q1a}
\mu_{t'+t}(Q_1) + \mu_{t'+t}(Q_2) \leq \frac{1}{\Gamma}\mu_{t'}(\overline{q}) +\frac{\Gamma-1}{\Gamma}\mu_{t'}(Q_2).
\end{equation}
Due to the order according to which the edges $\{q,v_i\}$ are queried, we have
\begin{equation} \label{eq:Q1b}
\mu_{t'}(Q_2) \leq \left(1 - \frac{t}{\deg(q)}\right) \mu_{t'}(\overline{q}) \leq \left(1 - \frac{t}{\Delta}\right) \mu_{t'}(\overline{q}).
\end{equation}
Note that $\mu_{t'}(\overline{q})\leq\frac{\Delta}{\Delta+1}\mu_{t'}$ since by assumption $q$ is $\frac{1}{\Delta+1}$-heavy in round $t'$.
Since $\mu_{t'+t} = \mu_{t'+t}(q) + \mu_{t'+t}(Q_1) + \mu_{t'+t}(Q_2)$, we get by~\eqref{eq:q1},~\eqref{eq:Q1a} and~\eqref{eq:Q1b}:
\[\mu_{t'+t}  \leq \left(\frac{1}{\Gamma} + \frac{\Gamma-1}{\Gamma} \frac{\Delta-t}{\Delta+1} \right)\mu_{t'}
            = \left(1 - \frac{\Gamma-1}{\Gamma} \frac{t  + 1}{(\Delta+1)} \right) \mu_{t'}
            \leq \left(1 - \frac{\Gamma-1}{\Gamma (\Delta+1)}\right)^{t+1} \mu_{t'},
\]
where the last inequality comes from $(1-x)^k \geq 1-xk$, for $k\geq 1$ and $x<1$.

Consider now a maximal sequence $S$ of rounds in which $q$ is $\frac{1}{\Delta+1}$-heavy and is not $\frac{1}{\Delta+1}$-heavy in the round that follows the sequence.
Note that Algorithm~\ref{str:graph_edges_fixed} cyclically queries the edges incident to $q$ in $S$.
Let $r_1'\leq\cdots\leq r_{b'}'$ be all rounds in $S$ having $q$ as an answer.
Denote $X=S\setminus\{r_1',\ldots,r_{b'}'\}$, the set of rounds in $S$ in which $q$ is not an answer.
Let $a=\lceil b'/\deg(q)\rceil$.
The lie counter of each vertex in $\overline{q}$ increases by at least $a-1$ and by at most $a$ times by executing $S$ --- we point out that this crucial property follows from the fact that the queries in the segment are applied to the edges incident to $q$ consecutively modulo $\deg(v)$.
Since $q$ is $\frac{1}{\Delta+1}$-heavy at the beginning of $S$ and is not $\frac{1}{\Delta+1}$-heavy right after $S$, the lie counter of $q$ increases by at least $a$ as a result of $S$.
Hence, $|X|\geq a$.
Partition $r_1',\ldots,r_{b'}'$ into a minimum number of segments of length at most $\deg(q)$ each, which leads to at most $a$ segments.
Thus, we can pair these segments with rounds in $X$.
For each such pair of at most $\deg(q)+1$ rounds we apply the amortized analysis performed above.
Note that this approach is valid since the amortized analysis is insensitive of the order of appearance of the queries in $X$ and the queries in $S\setminus X$.

Finally, suppose that there is a series of queries at the end of the strategy (a suffix) performed to edges incident to a $\frac{1}{\Delta+1}$-heavy vertex $q$ such that all replies point to $q$ and $q$ remains $\frac{1}{\Delta+1}$-heavy till the end of the strategy.
Note that in such a case $q$ is the target.
The vertex $q$ had the uniquely smallest lie counter just before those queries.
This in particular implies that the lie counter is strictly smaller than $\numlies$.
We artificially add a sequence of \emph{pseudo-queries}, each of which increments the lie counter of $q$ until it reaches $\numlies$.
This implies that the suffix of the search strategy now consists of a reply (which comes from a regular query or a pseudo-query) which does not point to $q$.
Thus, we use again the arguments from our amortized analysis: we can find a segment and pair with it the above mentioned query pointing away from $q$.

\subsection*{Proof of Theorem~\ref{th:edge_bounded}}

Similarly as in the case of vertex queries, the generic strategy in Algorithm~\ref{str:graph_edges_fixed} for edge queries and its corresponding bound for a fixed number of lies can be used to provide strong bounds for linearly bounded and probabilistic error models.

Let $\Gamma = 1+\frac{\Delta+1}{\Delta} \frac{\varepsilon}{1-\varepsilon} = \frac{1-r}{r} \cdot \frac{1}{\Delta}$. Denote $Q_{\min} = \frac{ \ln n}{\ln(1+ \frac{\Gamma-1}{\Gamma\Delta+1}) -  r \ln \Gamma}$. We run Algorithm~\ref{str:graph_edges_fixed} with bound $\numlies = Q_{\min} r$ and parameter $\Gamma$ set as just mentioned above.
Then, by Theorem~\ref{th:edges_queries_exact}, the length of the strategy is at most $\frac{1}{\ln(1+ \frac{\Gamma-1}{\Gamma\Delta+1})} \cdot \ln n +  \frac{\ln \Gamma}{\ln(1+ \frac{\Gamma-1}{\Gamma\Delta+1})} \cdot Q_{\min} r = Q_{\min} = L/r$. To conclude the proof, we bound
$$Q_{\min} =  \frac{\ln n}{F(\varepsilon)/(\Delta+1)+F(-\varepsilon/\Delta) \cdot \Delta/(\Delta+1)}=$$
(where $F(x) \defeq x + (1-x) \ln(1-x) = \sum_{i=2}^{\infty} \frac{x^i}{i(i-1)}$)
$$= \frac{\ln n}{\sum_{i=2}^{\infty} \frac{\varepsilon^i}{i(i-1)} \frac{\Delta^{i-1}+(-1)^{i}}{(\Delta+1)\Delta^{i-1}} } \le  \frac{\ln n}{\varepsilon^2/(2\Delta)} = 2\varepsilon^{-2}\Delta\ln n.$$

\subsection*{Proof of Theorem~\ref{th:edge-probabilistic}}

For edge queries, we use a two step approach: first, we repeatedly ask queries to boost their error rate from $\sim 1/2$ to below $1/(\Delta+1)$, and then use the linearly bounded error strategy.

As a first step, we show that for $p_0 = \frac{1}{\Delta+1}(1-\varepsilon_0)$, there exists a strategy that locates the target with high probability using $\bigo(\Delta \log n / \varepsilon_0^2)$ edge queries. Indeed, assume without loss of generality that $\varepsilon_0 < 1/2$.  We fix $\varepsilon_1 = \varepsilon_0 / (1+\sqrt{\frac32 \frac{\Delta+1}{\Delta} \ln \delta^{-1}/\ln n}) $ , and use Theorem~\ref{th:edge_bounded} with error rate $r_0 = \frac{1}{\Delta+1}(1-\varepsilon_1)$.  By Theorem~\ref{th:edge_bounded}, we obtain that the strategy length is $Q = 2 \varepsilon_1^{-2} \Delta \ln n  = \bigo(\Delta \varepsilon_0^{-2} (\log n + \log \delta^{-1}))$.
The expected number of lies is $\mathbb{E}[L] =p_0 \cdot Q$ and by the Chernoff bound,
$$\text{Pr}[L \ge r_0 \cdot Q ] \le \exp\left(-\frac13\left(\frac{r_0}{p_0}-1\right)^2 \cdot p_0 \cdot Q\right) \le \exp\left(-\frac23 \left(\frac{\varepsilon_0-\varepsilon_1}{\varepsilon_1}\right)^2 \cdot \frac{\Delta}{\Delta+1}  \ln n\right) = \delta.$$ 

We now observe that to achieve the error rate of $\frac12(1-\varepsilon)$, we can boost the query error rate to be smaller by repeating the same query multiple times and taking the majority answer. By repeating each query $P = \bigo( \log(2\Delta+2) \cdot \varepsilon^{-2})$ times, we get the correct answer with probability $1-p' = 1-\frac{1}{2} \cdot \frac{1}{\Delta+1}$, and as shown already, we only need $\bigo( \Delta (\log n + \log \delta^{-1}))$ queries with the error rate $p'$ to locate the target with probability at least $1-\delta$. Thus the claimed bound follows.

As an immediate corollary we obtain a very simple strategy for noisy binary search in an integer range of complexity $\bigo(\varepsilon^{-2}(\log n + \log \delta^{-1}))$ matching \cite{FeigeRPU94}.

\section{Application: Searching Unbounded Integer Ranges} \label{sec:unbounded}
Building on our generic strategies, we now obtain a general technique for searching an unbounded domain $\mathbb{N} = \{1,2,\ldots\}$ with comparison queries.
Here the measure of complexity is the dependency on the error rate (number of lies) and on $N$, the (initially unknown) position of the target. The main idea is to use Algorithms~\ref{str:graph_vertices_fixed} and~\ref{str:graph_edges_fixed}, tweaking the initial weight distribution. We fix the \emph{initial} weight of an integer $n$ to be $\mu_0(n) = n^{-2}$.
The total initial weight then equals $\pi^2/6=\Theta(1)$.
We provide the following bounds.\footnote{We note that the term \emph{ternary} refers to a model in which each query selects an integer $i$ and as a response receives information whether the target is smaller than $i$, equals $i$, or is greater than $i$.}

\begin{corollary}
\label{th:unbounded_fixed_vert}
There exists a strategy that finds an integer in an unbounded integer range ($\mathbb{N}$) using at most
\begin{itemize}
\item $\frac{\log \frac{\pi^2}{6} + 2 \log N + \numlies \log \Gamma }{\log\frac{2\Gamma}{\Gamma+1}}$ ternary queries, or
\item $\frac{\log \frac{\pi^2}{6}  + 2 \log N + \numlies \log \Gamma }{\log\frac{3\Gamma}{2\Gamma+1}}$ binary (comparison) queries,
\end{itemize}
where $N$ is the target, $\numlies$ is an upper bound on the number of (adversarial) lies and $\Gamma>1$ is an arbitrarily selected coefficient.
\end{corollary}
\begin{proof}
We use Algorithm~\ref{str:graph_vertices_fixed} for ternary queries; let the strategy length be $Q$.
By the proof of Theorem~\ref{th:vertex_queries_exact}, $\mu_Q\leq (\frac{2\Gamma}{\Gamma+1})^Q \cdot \frac{\pi^2}{6}$.
The final weight is at least $\mu_Q \ge N^{-2} \cdot \Gamma^{-\numlies}$, and the bound for ternary queries follows since the number of queries is at most $\log( \frac{\pi^2/6}{N^{-2}\Gamma^{-\numlies}} )/\log\frac{2\Gamma}{\Gamma+1}$.

The bound for binary queries is obtained analogously from Theorem~\ref{th:edges_queries_exact} (note that $\Delta=2$) since we apply Algorithm~\ref{str:graph_edges_fixed} for binary queries.
\end{proof}
Simply setting $\Gamma=2$ yields an $\bigo(\log N + \numlies)$ length strategy with comparison queries on unbounded integer domains with a fixed number of $\numlies$ lies.

We need to restate the linearly bounded error model in the case of unbounded domains since the Responder does not know a priori the length of the strategy.
We define this error model as follows: whenever the Questioner finds the target and thus declares the search to be completed after $t$ rounds, it is guaranteed that at most $rt$ lies have occurred throughout the search.
\begin{corollary} \label{th:unbounded_linearly-bounded}
For the linearly bounded error model with an error rate $r$ and an unbounded integer domain, there exists a strategy that finds the target integer $N$ in:
\begin{itemize}
\item $\bigo(\varepsilon^{-2} \log N)$ ternary queries when $r = \frac12(1-\varepsilon)$, or
\item $\bigo(\varepsilon^{-2} \log N )$ binary queries when $r = \frac13(1-\varepsilon)$.
\end{itemize}
\end{corollary}
\begin{proof}
Consider ternary queries.
We proceed analogously as in the proof of Theorem~\ref{th:vertex_lin_bounded}.
We have that the initial weight is $\pi^2/6$.
Run Algorithm~\ref{str:graph_vertices_fixed} until there is a single $n$ such that $\ell_n \le r \cdot Q$.
Any $Q$ such that $Q \ge \ln (\pi^2/6) / \ln\frac{2\Gamma}{\Gamma+1} + 2\ln N / \ln \frac{2\Gamma}{\Gamma+1} + L \ln \Gamma / \ln \frac{2\Gamma}{\Gamma+1}$ is an upper bound on the length of the strategy.
We thus get an upper bound
$$Q \le 2\varepsilon^{-2} (2\ln N + \bigo(1)).$$
The binary case follows in an analogous manner.
\end{proof}

We now proceed to show an algorithm for searching the ubounded integer range in the probabilistic error model. The challenge lies in the fact, that in all our previous algorithms we reduced the problem to the linearly bounded error model, and we could use an upper bound on the length of the strategy to select a proper relation between $p$ and $r$. However in this particular problem, the linearly bounded strategy could be arbitrarily long as $N \to \infty$.

\begin{figure}
\begin{center}
\begin{minipage}{.7\linewidth}
\begin{algorithm}[H]
\SetAlgoRefName{UNBOUNDED}
	\caption{Searching unbounded integer range in probabilistic error model.}
	\label{str:unbounded_noisy}
	$\delta' \gets \delta/2$\;
	\While{true}
	{
		$n \gets  1/\delta'$\;
		$t \gets \textsf{SEARCH}([n],\delta')$\;
		\If{$t \not= n$}
		{
			\Return $t$\;
		}
		\Else
		{
			$\delta' \gets \delta'^2$\;
		}
	}
\end{algorithm}
\end{minipage}
\end{center}
\end{figure}

Consider Algorithm~\ref{str:unbounded_noisy}. We assume that procedure  $\textsf{SEARCH}(I,\delta)$ implements a noisy binary search, that is, given $I \subset \mathbb{N}$, a parameter $\delta$ and a promise that the target $t$ is in $I$, the procedure returns $t \in I$ correctly with probability at least $1 - \delta$, assuming probabilistic error model (with error probability $p = \frac{1}{2}(1-\varepsilon)$). A strategy of length $\bigo(\varepsilon^{-2}(\log n+\log \delta^{-1}))$ is given in~\cite{FeigeRPU94}, so we can assume that $\textsf{SEARCH}(I,\delta)$ implements that particular strategy, but for our purposes any asymptotically optimal strategy suffices.

\begin{theorem}
\label{th:prob_bin_search_unbounded}
In the probabilistic error model, given integer domain $\mathbb{N}$ and unknown target integer $N$, Algorithm~\ref{str:unbounded_noisy} with probability at least $1 - \delta$ returns $N$ using at most $\bigo(\varepsilon^{-2}  (\log N + \log \delta^{-1}))$ binary queries for $p = \frac{1}{2}(1-\varepsilon)$. The strategy \emph{expected} number of queries is $\bigo(\varepsilon^{-2}  (\log N + \log \delta^{-1}))$ as well.
\end{theorem}
\begin{proof}
We first show that the algorithm is correct with probability at least $1-\delta$. Observe that the algorithm does a sequence of searches over increasing range of integer ranges $[n_0], [n_1], \ldots$, where $n_i =  1/\delta_i $ and $\delta_i = (\delta/2)^{2^i}$. If $\textsc{SEARCH}([n],\delta')$ is called with $N \ge n$, the returned value is $n$ with probability at least $1-\delta'$, which continues the main loop of the algorithm. If $N < n$, then with probability at least $1-\delta'$ returned value is $N$, which correctly stops the strategy. By union bound, failure probability of the algorithm is upperbounded by the sum of failure probabilities of all calls to $\textsc{SEARCH}$, which is $\sum_{i=0}^{\infty} \delta_i = \sum_{i=0}^{\infty}  (\delta_0)^{2^i} < \sum_{i=0}^{\infty}  (\delta_0)^{i}  \le 2 \delta_0 = \delta$.

We now bound the expected number of queries. Let $n_j > N$ be chosen such that $j$ is minimal. In the desired execution of algorithm, the last called search is on $[n_j]$, and it follows that either $j=0$ and then $n_j = 1/\delta_0$, or $n_j \le N^2+1$. Thus $\log n_j  = \bigo( \log N + \log 1/\delta)$. Let $C$ be a constant such that the number of queries performed by a call to $\textsc{SEARCH}([n],\delta')$ is upperbounded by $C \cdot (\log n + \log 1/\delta')$.

The expected number of queries is upperbounded by
\begin{align*}
\mathbb{E}[Q] &\le \sum_{i \le j} C \cdot (\log n_i + \log 1/\delta_i) + \sum_{i > j} \delta_{i-1} \cdot C \cdot (\log n_i + \log 1/\delta_i) \\
&\le 2 C\sum_{i \le j}\log n_i + 2 C\sum_{i > j} \delta_{i-1} \log n_i \\
 &\le 4 C \log n_j + 2 C  \sum_{i>0} 2^i  \cdot (\delta_j)^{2^{i-1}} \log n_j  \\
&\le 2 C \left(2 + \sum_{i>0} 2^{i-2^{i-1}} \right)  \log n_j  \\
&\le 10 C \log n_j \\
&= \bigo(\log N + \log 1/\delta),
\end{align*}
where we have used that $\delta_j \le 1/2$.
\end{proof}

\section{Application: Edge Queries in the Prefix-Bounded Model} \label{sec:edge-prefix-bounded}
The model of prefix-bounded errors can be somewhat seen as lying in-between the adversarial linearly bounded and the non-adversarial probabilistic.
It is reflected e.g. in the fact that in binary search the `feasibility' threshold for $r$ changes from $\frac{1}{3}$ in the linearly bounded to $\frac{1}{2}$ in the prefix-bounded model.
We utilize the ideas from~\cite{BorgstromK93} more carefully, adapting the approach to edge queries in general graphs and the prefix-bounded error model.
It turns out that the feasibility threshold for $r$ can be pushed from $\frac{1}{\Delta+1}$ to $\frac{1}{\Delta}$ in this case, while keeping the $\log n$ dependency on the graph size.\footnote{We note that for any $r=1/2(1-\varepsilon)$ there exists a strategy of exponential length $(1/\varepsilon)^{\bigo(\Delta \log n)}$, following from straightforward simulation of error-less strategy by repeating queries.}

For the virtual advance technique that we utilize, in addition to the lie counter $\ell_v$ of a vertex $v$ that we used so far, we introduce a \emph{virtual lie counter}, denoted by $\virtLies{v}$, that is maintained by our strategy given in Algorithm~\ref{str:prunning}.
Whenever a query is made to an edge $\{u,v\}$ and the reply is $u$, then the virtual lie counter of $u$ is incremented by the strategy (note that this reply results in incrementing $\ell_v$ but $\ell_u$ remains the same).
We extend the notation by introducing a \emph{virtual potential} $\virtPhi(v)$ for each node $v$, $\virtPhi(v)=\virtPhi_0(v)\cdot\Gamma^{-(\ell_v+\virtLies{v})}$, where $\virtPhi_0(v)$ is the initial potential (in Algorithm~\ref{str:prunning}, $\virtPhi_0(v)=1$ for all $v$).
Consequently, we define for each edge $e$, $\virtPsi(e)=\sum_{u\in V}\virtPhi(u)\cdot d(e,u)$.
The strategy relies on two constants $\Gamma$ and $H$ that we select while stating our lemmas below.
The values of $C$ and $D$ in Algorithm~\ref{str:prunning} computed in round $t$ of the strategy are denoted during analysis by $C_t$ and $D_t$, respectively.
The goal of Algorithm~\ref{str:prunning} is to trim down the set of potential targets to at most $\bigo(\Delta/\varepsilon)$, where $r=\frac{1}{\Delta}(1-\varepsilon)$.

\begin{center}
\begin{minipage}{.7\linewidth}
\begin{algorithm}[H]
\SetAlgoRefName{PRUNING}
	\caption{Edge queries for the prefix-bounded model}
	\label{str:prunning}
	\For{ $v \in V$ }
	{
		$\virtPhi(v) \gets 1$\;
		$\ell_v \gets 0$\;
		$\virtLies{v} \gets 0$
	}
	\Do{$|C|>1$}
	{
		$e \gets \arg \min_{x \in E} \virtPsi(x)$\;
		query $e$ with an answer $w$\;
		\For{all nodes $u$ not compatible with the answer}
		{
			$\ell_u \gets \ell_u + 1$\;
			$\virtPhi(u) \gets \virtPhi(u)/\Gamma$
		}
		$\virtLies{w} \gets \virtLies{w}+1$\;
		$\virtPhi(w) \gets \virtPhi(w)/\Gamma$\;
                $D \gets \{ u : \virtLies{u} \ge t/H\}$\;
                $C \gets \{ u \in V \setminus D : \ell_u \le r \cdot t\}$.
	}
	\Return $C\cup D$
\end{algorithm}
\end{minipage}
\end{center}

\begin{theorem}
\label{th:edge_prefix_bounded}
Algorithm~\ref{str:prunning} with $H=2\Delta \varepsilon^{-1}$ and $\Gamma = 1 + \frac{\Delta}{2(\Delta-1)} \varepsilon$ in $Q_{0} = \bigo(\Delta \varepsilon^{-2} \log n )$ edge queries returns a set $D$ of size $\bigo(\Delta \varepsilon^{-1})$ of possible target candidates in the prefix-bounded error model with $r=\frac{1}{\Delta}(1-\varepsilon)$.
\end{theorem}
\begin{proof}
Denote $\varepsilon' = \varepsilon/2$ and $r' = \frac{1}{\Delta}(1-\varepsilon')$.
Note that $\Gamma = 1 + \frac{\Delta}{\Delta-1} \cdot \varepsilon'$ and $H = 1/(r'-r) = \Delta/\varepsilon'$.
We prove that in at most
\begin{equation} \label{eq:Qmax}
Q_{0} = \left(8(\Delta-1)\varepsilon^{-2} + \bigo(\varepsilon^{-1})\right) \ln n
\end{equation}
edge queries Algorithm \ref{str:prunning} terminates.

If an edge $e=\{u,v\}$ is a $1$-median with respect to $\virtPsi$ and $\deg(u)>1$, where $u$ is the reply to the query in round $t+1$, then by Lemma~\ref{lem:degBound} applied to the minimizer $e$ of $\virtPsi$,
\begin{equation} \label{eq:virtPhi}
\virtPhi_t(\overline{N(e,u)}) \geq \frac{1}{\deg(u)} ( \virtPhi_t - \virtPhi_{t}(u) ).
\end{equation}
Note that if $\deg(u)=1$, then $\virtPhi_t(\overline{N(e,u)})=\virtPhi_t(V\setminus\{u\})=\virtPhi_t-\virtPhi_t(u)$, which implies that in this case~\eqref{eq:virtPhi} also holds.
Hence we obtain from~\eqref{eq:virtPhi}:
\[\virtPhi_t(\overline{N(e,u)} \cup \{u\}) \geq \frac{1}{\Delta} \virtPhi_t.\]
Thus, in each round, the decrease in the virtual potential is as follows:
\begin{align*}
\virtPhi_{t+1}  &= \virtPhi_t(N(e,u)\setminus\{u\}) + \frac{1}{\Gamma}\virtPhi_t(\overline{N(e,u)}\cup\{u\}) \\
                &= \virtPhi_t - \frac{\Gamma-1}{\Gamma}\virtPhi_t(\overline{N(e,u)}\cup\{u\}) \\
                &\leq \left( \frac{1}{\Gamma\Delta} + \frac{\Delta-1}{\Delta}\right)\virtPhi_t.
\end{align*}
Since the initial virtual potential is $\virtPhi_0=n$, this implies
\begin{equation} \label{eq:Qmax0}
\virtPhi_{Q_{0}} \leq n\cdot\left( \frac{1}{\Gamma\Delta} + \frac{\Delta-1}{\Delta}\right)^{Q_{0}}.
\end{equation}
Observe that
\begin{equation}
\left(\frac{1}{\Gamma\Delta}+\frac{\Delta-1}{\Delta}\right)\cdot\Gamma^{r}\cdot \Gamma^{1/H} = 1 - \frac{\varepsilon^2}{8(\Delta-1)} + \bigo(\varepsilon^3).
\end{equation}
Thus, by~\eqref{eq:Qmax} and~\eqref{eq:Qmax0}, the total virtual potential after $Q_{0}$ queries is at most
\begin{equation} \label{eq:Q-Gamma}
\virtPhi_{Q_{0}} \le n \cdot \left(\frac{1 - \frac{\varepsilon^2}{8(\Delta-1)}+\bigo(\varepsilon^3)}{\Gamma^{r+1/H}}\right)^{Q_{0}}
                = \Gamma^{-{Q_{0}}(r+1/H)}.
\end{equation}

Denote for brevity $D'=D_{Q_{0}}$.
Since we had $Q_{0}$ rounds and in each round the virtual potential of exactly one vertex increases, there are at most $H$ discarded vertices in $D'$.
For all other vertices in $V \setminus D'$, the virtual lie counter does not exceed $Q_{0}/H$ according to Algorithm~\ref{str:prunning}.
Thus, by~\eqref{eq:Q-Gamma},
\[\Phi_{Q_{0}}(V\setminus D') = \sum_{v\in V\setminus D'} \Phi_{Q_{0}}(v) \leq \Gamma^{Q_{0}/H} \cdot \sum_{v\in V\setminus D'} \virtPhi_{Q_{0}}(v) \leq \Gamma^{Q_{0}/H} \cdot \virtPhi_{Q_{0}} \leq \Gamma^{-Q_{0}r}.\]
This means that there is at most one vertex $v \in V \setminus D'$ such that $\ell_v \le r \cdot Q_{0}$.
Thus, Algorithm~\ref{str:prunning} indeed terminates in at most $Q_{0}$ rounds.
Additionally, in any round $t$, $|D_t| \le H$, which proves our claim.
\end{proof}

\begin{corollary}
\label{cor:range_prefix_bounded}
In the prefix-bounded error model with $r=\frac{1}{2}(1-\varepsilon)$, the target in an integer domain can be found in $\bigo( \varepsilon^{-4 }\log n )$ binary queries.
\end{corollary}
\begin{proof}
We first use the strategy described from Theorem~\ref{th:edge_prefix_bounded} to reduce, in $Q_{0} = \bigo(\varepsilon^{-2} \log n)$ rounds, the set of potential targets to $C\cup D$, where $|C|=1$ and $|D| = \bigo(\varepsilon^{-1})$.
In case of no further errors, $C\cup D$ can be then reduced in $Q' \le 1+\log_2 |D|$ queries to a single target.
The final strategy can be simulated as described in \cite{aslam1995noise}, giving the total strategy length of
$\bigo( Q_{0} \cdot \frac{1}{1-2r} \cdot (\frac{1}{1-r})^{Q'})$. Since $\frac{1}{1-r} \le 2$ and $\frac{1}{1-2r} = \varepsilon^{-1}$, this results in $\bigo( \varepsilon^{-4}\log n)$ edge queries, as claimed.
\end{proof}

We note that the simulation argument from \cite{aslam1995noise} requires that for any queried edge $e = \{u,v\}$ and a set of potential targets $D$, it holds $D \subseteq N_{<}(e,v) \cup N_{<}(e,w)$. This is always true e.g. in bipartite graphs regardless of $D$.

To obtain our results for the prefix-bounded error model and general graphs, we use the `trimming' phase provided by Theorem~\ref{th:edge_prefix_bounded} which is then followed by a simulation argument.
The latter requires an error-less strategy whose queries are then repeated e.g. for majority testing.
The theorem below provides such an edge search strategy for an arbitrary graph.

\begin{theorem}
\label{th:edge_errorless}
There exists a strategy that in absence of errors finds the target in at most $\frac{\log(n/\Delta)}{\log(\Delta/(\Delta-1))} + \Delta$ edge queries in any $n$-node graph of max-degree $\Delta$.
\end{theorem}
\begin{proof}
We use Algorithm~\ref{str:graph_edges_fixed} with a simplification of taking $\Gamma \to \infty$.
Thus we have $\Phi(v) \in \{0,1\}$ and these occur for $\ell_v > 0$ and $\ell_v = 0$, respectively.
Let $S_t$ be the set of potential targets after $t$ queries.
Note that $S_t=\{v \mid \Phi_t(v)=1\}$.
By Lemma~\ref{lem:degBound}, it follows that in any step querying an edge $e_q$ with an answer $q$, the discarded set of targets satisfies
$$|S_t \cap \overline{N(e_q,q)}| \geq  \frac{1}{\Delta} |  S_t \setminus \{q\} | \geq \frac{1}{\Delta} (|S_t|-1).$$
Thus, $|S_{t+1}| = |S_t \cap N(e_q,q)| \leq |S_t| - \frac{1}{\Delta} (|S_t|-1)$. From $|S_{t+1}|-1 \le (|S_t|-1) \cdot \frac{\Delta-1}{\Delta}$ we deduce that it takes at most $\left\lceil\frac{\log(n/\Delta)}{\log(\Delta/(\Delta-1))}\right\rceil$ queries to reduce target set to size at most $\Delta$, and then another $\Delta-1$ queries to reduce it to a single target.
\end{proof}

\begin{theorem}
\label{cor:edge_prefix_bounded}
In the prefix-bounded error model with $r = \frac{1}{\Delta}(1-\varepsilon)$, the target can be found in $\varepsilon^{-\bigo(\Delta)} \log n $ edge queries in general graphs.
\end{theorem}
\begin{proof}
Suppose that $D$ is a set of potential targets, i.e., the target $v$ belongs to $D$.
By Theorem~\ref{th:edge_errorless}, there exists a strategy (for the error-less case) with at most $Q' \le \frac{\log (|D|/\Delta)}{\log(\Delta/(\Delta-1))}+\Delta$ edge queries that finds the target $v$.

First assume that $\Delta \ge 3$.
It follows immediately from the simulation argument from~\cite{Pelc89} (in which one repeats multiple times a query of another strategy taking majority answer in each case --- here we use the error-less strategy of length $Q'$ from Theorem~\ref{th:edge_errorless}) that there exists a strategy terminating in $\bigo(Q_{0} \cdot (1/(1-2r))^{Q'}) = \bigo(\varepsilon^{-2}\Delta \log n) \cdot (1/(1-2r))^{\bigo(\Delta \log \varepsilon^{-1})}$ 
edge queries, where $Q_0$ is the length of the strategy produced by Algorithm~\ref{str:prunning}.
Note that the value of $Q_{0}$ comes from Theorem~\ref{th:edge_prefix_bounded}. Since $\Delta \ge 3$, $1/(1-2r) \le 3$, the claimed bound immediately follows.

For $\Delta=2$, the only cases not covered by Corollary~\ref{cor:range_prefix_bounded} are in fact odd-length cycles. We deal with them as follows. The initial sequence of queries is done as previously --- by executing Algorithm~\ref{str:prunning}, reducing the set of potential targets to $D$ at the cost of $Q_{0}$ rounds. We now observe, that for any edge $e=\{u,v\}$, there is a \emph{single} vertex $v_e$ such that $d(u,v_e) = d(v,v_e)$.
Thus we can consider the following error-less strategy applied to the set of potential targets $D$: query edges according to an error-less edge strategy (as in Theorem~\ref{th:edge_errorless}) and for each queried edge $e$, \emph{discard} the vertex $v_e$ from the set of potential targets.
At the end of this strategy, reintroduce all discarded vertices.
This strategy can be simulated as in  \cite{aslam1995noise}, since we always make sure to maintain the property of properly bisecting the set of targets.
Thus, our initial $Q_0 = \bigo(\varepsilon^{-2} \log n)$ rounds and $D_0=|D| = \bigo(\varepsilon^{-1})$ targets give that this strategy has length $Q_1 = \bigo(Q_0 \cdot \varepsilon^{-1} \cdot 2^{\log_2 D_0}) = \bigo(Q_0 D_0 \varepsilon^{-1})$ and results in $D_1 \le 2 + \log_2 D_0$ targets.
Iterating this procedure would give us a strategy of length $\varepsilon^{-\bigo(\log^* \varepsilon^{-1})} \log n$.
To improve its length by getting rid of the non-constant exponent, denote by $E_0$ set of edges queried during the transition from $D_0$ to $D_1$. Since the strategy is basically a binary search, there are $\bigo(1)$ pairs of edges in $E_0$ that share an endpoint, and there are $\bigo(1)$ of pairs of vertices in $D_1$ that share an edge. Thus, the querying strategy of reducing $D_1$ to $D_2$ can always, except for $\bigo(1)$ queries, choose an edge $e$ to be queried so that $v_e \not\in D_1$. Thus $D_2 = \bigo(1)$, and the proof concludes.
\end{proof}

\begin{corollary}
\label{th:prefix_bound_unbounded}
In the prefix-bounded error model with $r = \frac12(1-\varepsilon)$, $0<\varepsilon\leq 1$, the target integer in an unbounded integer domain can be found in $\bigo(\varepsilon^{-4} \log N)$ binary queries.
\end{corollary}
\begin{proof}
Set $s=2$ and proceed first with the filtering technique by executing Algorithm~\ref{str:prunning} with $H = \frac{4}{\varepsilon}$ and $\Gamma = 1 + \varepsilon$. Following the proof of Theorem~\ref{th:edge_prefix_bounded}, we observe that a single query reduces the adjusted potential in $V \setminus D$ by a factor $1 + \frac{\varepsilon^2}8 - \bigo(\varepsilon^3)$. After $Q_{0} = (8 \varepsilon^{-2} + \bigo(\varepsilon^{-1}))\cdot \ln (\zeta(s) \cdot N^s)$ queries the potential of the vertices in $V\setminus D$ is reduced from $\zeta(s)$ to $N^{-s}$, meaning that the set $C_{Q_{0}}$ has only one vertex. We apply Corollary~\ref{cor:range_prefix_bounded} to $D_{Q_0} \cup C_{Q_{0}}$, which is of size at most $4 \varepsilon^{-1}+1$.
\end{proof}

\section{Summary of Results}

\renewcommand{\arraystretch}{1.5}

We conclude by grouping all bounds we have obtained in three tables below.
In each case $\varepsilon$ is the relative difference between the assumed upper bound for $r$ or $p$ and this value itself, i.e. in the context of $r<r_{\max}$ (or $p<p_{\max}$ respectively) it satisfies $r = (1-\varepsilon)r_{\max}$ (respectively $p = (1-\varepsilon)p_{\max}$).
For the probabilistic model, $\delta$ is the probability threshold, i.e., the target must be located with probability at least $1-\delta$.
Our results are compared with the best ones known to date.
Keep in mind that for $p=\frac{1}{2}(1-\varepsilon)$, it holds $1-H(p) = \Theta(\varepsilon^2)$.

%Table~\ref{tab:graphs_unbounded} gives the graph-theoretic results, which are mostly new.
\begin{table}[ht!]
\caption{Query complexity in general graphs.}
\label{tab:graphs_unbounded}
{
\footnotesize
\begin{center}
\begin{tabular}{cccccc}
\toprule
\emph{Model}    &    \emph{Queries}  & \emph{Regime}       & \emph{Previous result} & \emph{Our result} &                                     \\
\midrule
\multirow{2}{*}{fixed} & vertex & & - & 
 $\bigo(\log n + \numlies)$ & \eqref{th:vertex_queries_exact}\\
\cline{2-6}
& edge & & - & $\bigo(\Delta(\log n + \numlies))$ & \eqref{th:edges_queries_exact}
\\
\midrule
\multirow{2}{*}{\makecell{linearly\\bounded}} & vertex & $r < \frac{1}{2}$ & - & 
$\frac{\log_2 n}{1-H(r)}$ &  \eqref{th:vertex_lin_bounded}\\
\cline{2-6}
& edge &  $r=\frac{1}{\Delta+1}(1-\varepsilon)$ & - & $2 \varepsilon^{-2} \Delta \ln n $ & \eqref{th:edge_bounded}
\\
\midrule
\makecell{prefix-\\bounded} & edge  & $r = \frac{1}{\Delta}(1-\varepsilon)$ & - & $\left(\frac{1}{\varepsilon}\right)^{\bigo(\Delta)} \log n $ &  \eqref{cor:edge_prefix_bounded} \\
\midrule
\multirow{2}{*}{\makecell{probabi-\\listic}}   &  vertex     & $p<\frac12$  & \makecell[c]{ $\frac{\log_2 n}{1 - H(p)} +$\\$+\frac{1}{1-H(p)}\bigo(\frac1C\log n + C^2 \log \delta^{-1})$ \\ $C=\max(1,(1/2-p)\sqrt{\log \log n})$ \cite{Emamjomeh-Zadeh:2015aa}}   & 
 \makecell[c]{$\frac{\log_2 n + o(n) + \widetilde\bigo(\log \delta^{-1})}{1 - H(p)} $ } & \eqref{th:vertex-probabilistic}\\
\cline{2-6}
& edge & $p=\frac12(1-\varepsilon)$  & - & \makecell[r]{$\bigo(\varepsilon^{-2} \Delta \log \Delta $ \\ $(\log n + \log \delta^{-1}))$} & \eqref{th:edge-probabilistic}
\\
\bottomrule
\end{tabular}
\end{center}
}
\medskip
\end{table}

\begin{table}[ht!]
\caption{Query complexity in linearly ordered integer ranges of length $n$ with comparison queries, i.e., the generalizations of the classical binary search (equivalent to the edge-query model in paths).}
\label{tab:results_bounded}
{
\footnotesize
\begin{center}
\begin{tabular}{ccccccc}
\toprule
\emph{Model}&    \emph{Queries}    &   \emph{Regime}      & \emph{Previous result} & & \emph{Our result} & \\
\midrule
fixed & binary  &  &$\bigo(\log n + \numlies)$ & \makecell{\cite{Aigner96},\\ \cite{Long92}} & $\bigo(\log n+\numlies)$ & \eqref{th:edges_queries_exact} \\
\midrule
\multirow{2}{*}{\makecell{linearly\\bounded}} & binary &  $r=\frac{1}{3}(1-\varepsilon)$ & $8\varepsilon^{-2}\log_2n$ &  \cite{DhagatGW92} & 
$4  \varepsilon^{-2} \ln n$ & \eqref{th:edge_bounded}\\
\cline{2-7}
& ternary & $r <\frac12$ & -  & & $\frac{\log_2 n}{1-H(r)}$ &  \eqref{th:vertex_lin_bounded} \\
\midrule
\makecell{prefix-\\bounded}  & binary & $r=\frac{1}{2}(1-\varepsilon)$ & $\bigo(\textup{poly}(\varepsilon^{-1})\log n)$ & \cite{BorgstromK93}  & 
$\bigo(\varepsilon^{-4}{\log n })$ & \eqref{cor:range_prefix_bounded}
 \\
\midrule
probabilistic    & binary &    $p<\frac12$   & \makecell{$\frac{\log_2 n}{1 - H(p)} +$\\ $+\bigo(\frac{\log \delta^{-1} + \log \log n}{1 - H(p)} )$} & 
 \cite{Ben-OrH08}  &  
 $\bigo(\frac{\log n+ \log \delta^{-1}}{1-H(p)})$ & \eqref{th:edge-probabilistic}\\
\bottomrule
\end{tabular}
\end{center}
}
\medskip
\end{table}

\begin{table}[ht!]
\caption{Query complexity in unbounded integer domain $\mathbb{N} = \{1,2,\ldots\}$ (here $N$ is the value of the unknown target integer).}
\label{tab:results_unbounded}
{
\small
\begin{center}
\begin{tabular}{ccccccc}
\toprule
\emph{Model}   &    \emph{Queries}  & \emph{Regime}        & \emph{Previous result} & & \emph{Our result} & \\
\midrule
fixed & binary & & - & & $\bigo(\log N + \numlies)$ & \eqref{th:unbounded_fixed_vert}\\
\midrule
\multirow{2}{*}{\makecell{linearly\\bounded}} & binary & $r=\frac13(1-\varepsilon)$ & -  & &
$\bigo(\varepsilon^{-2}\log N)$ &  \eqref{th:unbounded_linearly-bounded}\\
\cline{2-7}
& ternary &  $r=\frac12(1-\varepsilon)$ & -  & &$\bigo(\varepsilon^{-2}\log N)$ & \eqref{th:unbounded_linearly-bounded}\\
\midrule
\makecell{prefix-\\bounded}& binary & $r = \frac{1}{2}(1-\varepsilon)$ & $\bigo(\varepsilon^{-3}(N\log^2N)^{\log_2\frac{2}{1+\varepsilon}})$ &  \cite{AslamD91} &   $\bigo(\varepsilon^{-4}\log N)$ &  \eqref{th:prefix_bound_unbounded}\\
\midrule
probabilistic & binary &    $p=\frac12(1-\varepsilon)$     & $\bigo(\text{poly}(\varepsilon^{-1})\log N \log{\delta^{-1}})$ & \cite{aslam1995noise} & $\bigo(\varepsilon^{-2} (\log N + \log \delta^{-1}))$  & \eqref{th:prob_bin_search_unbounded} \\
\bottomrule
\end{tabular}
\end{center}
}
\medskip
\end{table}

\end{document}